\documentclass[journal,12pt,draftcls,onecolumn]{IEEEtran}
\IEEEoverridecommandlockouts

\usepackage[utf8]{inputenc}
\usepackage{graphicx}
\usepackage{amssymb, amsmath,amsthm, amsfonts}
\usepackage{tikz}
\usepackage{pgfplots}
\pgfplotsset{compat=1.17}
\usepgfplotslibrary{groupplots} 
\usetikzlibrary{pgfplots.groupplots}
\usepackage{thmtools,thm-restate}

\usepgfplotslibrary{fillbetween}
\usepackage{verbatim}
\usepackage{pifont}
\usetikzlibrary{patterns}
\title{Coded Shotgun Sequencing}
\author{Aditya Narayan}
\newcommand*\Y{\mathcal{Y}}
\usepackage[ruled,vlined,noend]{algorithm2e} 
\usepackage{float}

\usepackage[numbers,sort&compress]{natbib}

\newtheorem{theorem}{Theorem}

\newtheorem{lemma}{Lemma}

\newtheorem{defn}{Definition}

\newcommand{\aln}[1]{\begin{align*}#1\end{align*}}
\newcommand{\al}[1]{\begin{align}#1\end{align}}

\newcommand{\1}{{\bf 1}}

\newcommand{\ep}{\epsilon}

\newcommand{\Ber}{\text{Ber}}

\renewcommand{\P}{{\mathcal P}}
\newcommand{\E}{{\mathcal E}}
\newcommand{\B}{{\mathcal B}}

\newif\ifdraft
\drafttrue

\begin{document}

\title{Coded Shotgun Sequencing
}
\author{
Aditya~Narayan~Ravi, Alireza~Vahid, Ilan~Shomorony
\thanks{Aditya Narayan Ravi and Ilan Shomorony are with the Electrical and Computer Engineering Department of the University of Illinois, Urbana-Champaign, IL, USA. Email: {\sffamily anravi2@illinois.edu,ilans@illinois.edu}.}
\thanks{Alireza Vahid is with the Electrical Engineering Department of the University of Colorado Denver, Denver, CO, USA. Email: {\sffamily alireza.vahid@ucdenver.edu}.}
}


 \maketitle


\begin{abstract}
Most DNA sequencing technologies are based on the shotgun paradigm: many short reads are obtained from random unknown locations in the DNA sequence. 
A fundamental question, studied in~\cite{MotahariDNA}, is what read length and coverage depth (i.e., the total number of reads) are needed to guarantee reliable sequence reconstruction. Motivated by DNA-based storage, we study the coded version of this problem; i.e., the scenario where the DNA molecule being sequenced is a codeword from a predefined codebook. 
Our main result is an exact characterization of the capacity of the resulting \emph{shotgun sequencing channel} as a function of the read length and coverage depth. In particular, our results imply that, while in the uncoded case, $O(n)$ reads of length greater than $2 \log n$ are needed for reliable reconstruction of a length-$n$ binary sequence, in the coded case, only $O(n/\log n )$ reads of length greater than $\log n$ are needed for the capacity to be arbitrarily close to $1$.
\end{abstract}



\section{Introduction}
\label{Section:Introduction}

Over the last decade, advances in 
DNA sequencing technologies have driven down the time and cost of acquiring 
biological data tremendously.
At the heart of this sequencing revolution was the development of \emph{high-throughput shotgun sequencing} platforms.
Rather than attempting to read a long DNA molecule from beginning to end, these platforms extract a large number of short reads from random locations of the target DNA sequence (e.g., the genome of an organism), in a massively parallel fashion.
Sequencing must then be followed by an \emph{assembly} step, where the reads are merged together based on regions of overlap with the intention of reconstructing the original DNA sequence.

In the context of this shotgun sequencing pipeline, it is natural to ask when it is possible, from an information-theoretic standpoint, to reconstruct a sequence from a random set of its substrings. 
More precisely, suppose we observe $K$ random reads (i.e., substrings) of length $L$ from an unknown length-$n$ sequence $x^n$.
What conditions on $x^n$, $K$ and $L$ guarantee that $x^n$ can be reliably reconstructed from the observed reads?
This problem was first studied from an information-theoretic point of view by Motahari et al.~\cite{MotahariDNA}. 
The authors considered the asymptotic regime where $n \to \infty$ and the read length $L$ scales as
\al{ \label{eq:readlength}
L = \bar{L} \log n,
}
for a constant $\bar L$.
They also defined $c = \frac{KL}{n}$ to be the \emph{coverage depth};
i.e., the average number of times each symbol in $x^n$ is sequenced.
This appropriate scaling of the read length allowed the authors of \cite{MotahariDNA} to show a surprising critical phenomenon: 
if $x^n$ is an i.i.d.~$\Ber(1/2)$ sequence, when $\bar L < 2$, reconstruction is impossible for any coverage depth $c$, but if $\bar L > 2$, reconstruction is possible as long as the coverage depth is at least the Lander-Waterman coverage
\al{
c_{LW} = \ln\left( \frac{n}{\ep} \right).
}
The Lander-Waterman coverage \cite{LanderWaterman} is the minimum coverage needed to guarantee that all symbols in $x^n$ are sequenced at least once with probability $1-\ep$.
The result in \cite{MotahariDNA} established a \emph{feasibility region} for the shotgun sequencing problem, illustrated in blue in Figure~\ref{fig:regions}.
Notice that the number of reads required is linear in $n$ since
\aln{
K = \frac{n}{L} \cdot c_{LW} = 
\frac{n}{L} \ln \left( \frac{n}{\ep} \right) = \frac{n \ln \left( n /\ep \right)}{\bar L \log n} = \Theta(n).
}

One key aspect about the framework studied in \cite{MotahariDNA} is that the sequence $x^n$ is chosen ``by nature'' (which can be modeled as a random process as in \cite{MotahariDNA} or as an unknown deterministic sequence as later done in \cite{BBT,shomorony2016fundamental}).
However, in recent years, significant advances in DNA synthesis technologies have enabled the idea of storing data in DNA,
and several groups demonstrated
 working DNA-based storage systems~\cite{church_next-generation_2012,goldman_towards_2013,grass_robust_2015,yazdi_rewritable_2015,erlich_dna_2016,organick_scaling_2017,antkowiak_low_2020}.
In these systems, information was 
encoded into DNA molecules via state-of-the-art synthesis techniques, and later retrieved via sequencing.
This emerging technology motivates the following question: How do the fundamental limits of shotgun sequencing from \cite{MotahariDNA} change in the coded setting where $x^n$ is chosen from a codebook?

\begin{figure}
\centering 
\begin{tikzpicture}
\begin{axis}[grid=none,
axis lines=middle,
extra x tick style={ticklabel style={fill=white,font=\small}},
extra x tick labels={$1$,$2$},
extra y tick labels={$\Omega\left(\frac{n}{\log{n}}\right)$,$\Theta\left(n\right)$},
extra y tick style={ticklabel style={fill=white,font=\small}},
xmin=0,xmax=4,xtick={0,1,2,3},xticklabels={\empty},
ymin=0,ymax=3,ytick={0,1,2},yticklabels={\empty},
extra x ticks={1,2}, extra y ticks={1,2},
xlabel=\(\Bar{L}\),ylabel=\(K\),
samples=200]
\addplot[dashed,very thick,color=red,name path=A] coordinates{(1,1) 
(6,1)};
\addplot[dashed,very thick,color=red,name path=B] coordinates{(1,0) 
(1,6)};
\addplot[dashed,very thick,color=blue,name path=C] coordinates{(2,2) 
(6,2)};
\addplot[dashed,very thick,color=blue,name path=D] coordinates{(2,0)(2,6)};
\addplot [red!30] fill between [
        of=A and B,
    ];
\addplot [blue!30] fill between [
        of=C and D,
];
\end{axis}
\end{tikzpicture}
\caption{The blue region describes a feasible region where the normalized read length $\bar L$ and the number of reads $K$ satisfies conditions needed for perfect sequence reconstruction in the uncoded setting \cite{MotahariDNA}.
In the coded setting studied in this paper, the requirements for the capacity to be $1$ are significantly less stringent: $\bar L > 1$ and $K$ growing faster than $n/\log n$ suffices.
\label{fig:regions}}
\end{figure}
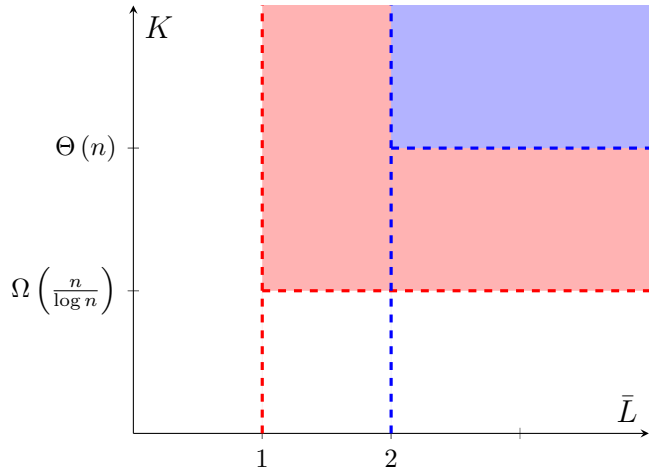


Motivated by this question, in this paper we introduce the 
Shotgun Sequencing Channel (SSC). 
As illustrated in Figure~\ref{fig:ChannelModel}(a), the channel input is a (binary) length-$n$ sequence $x^n$,
and the channel output are $K$ random reads of length $L$ from $x^n$.
Each read is assumed to be drawn independently and uniformly at random from $x^n$ and we consider the read length scaling in (\ref{eq:readlength}).
Notice that this is essentially the same setup as in \cite{MotahariDNA}, except that the ``genome'' $x^n$ is chosen from a codebook rather than decided by nature.
Our goal is to characterize the capacity of this channel.


\begin{figure}[ht!]
\centering
\includegraphics[clip,width=0.85\linewidth]{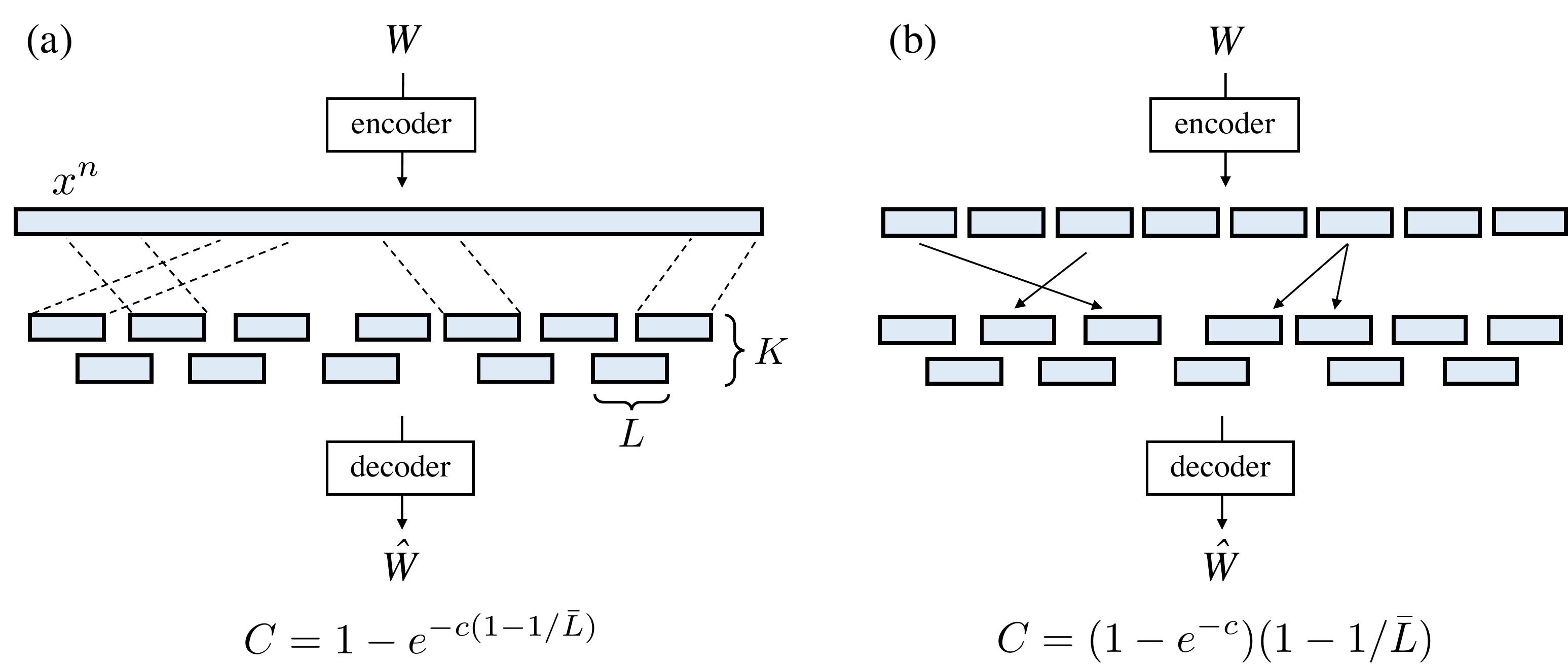}
\caption{Comparison between the (a) Shotgun Sequencing Channel (SSC) and the (b) Shuffling-Sampling channel from \cite{DNAStorageISIT} and the corresponding capacity expressions. 
The input to the SSC is a single (binary) string $x^n$ and the output are $K$ random substrings of length $L$.
In the Shuffling-Sampling channel, the input are $M$ strings of length $L$, which are sampled with replacement to produce the channel output.
Both capacity expressions can be written in terms of the expected coverage depth $c$ and the normalized read length $\bar L$.
\label{fig:ChannelModel}}
\end{figure}


In order to build intuition it is worth considering the related setting of the \emph{shuffling-sampling channel} \cite{DNAStorageISIT}, illustrated in Figure~\ref{fig:ChannelModel}(b).
In this case the input are $M$ strings of length $L$, and the output are $K$ strings, each chosen uniformly at random from the set of input strings.
If we define the coverage depth for this setting as
$c = \frac{KL}{ML} = K/M$,
%
%
%
%
the result in \cite{DNAStorageISIT} implies that, for $\bar L > 1$, the capacity of this channel is
\begin{align} \label{eq:capshuff}
    C_{\text{shuf}} = \left(1-e^{-c}\right)\left(1-1/\Bar{L}\right),
\end{align}
and $C_{\rm shuf} = 0$ for $\bar L \leq 1$.
The term $(1-e^{-c})$ captures the loss due to unseen input strings and $(1-1/\bar{L})$ captures the loss due to the unordered nature of the output strings (which becomes more severe the shorter the strings are).

Intuitively, the capacity of the SSC should depend on $c$ and $\bar{L}$ in a similar way as in (\ref{eq:capshuff}).
The expected fraction of symbols in $x^n$ that are read at least once can be shown to be $1-e^{-c}$, which provides an upper bound to the capacity of the SSC.
But it is not clear a priori which of the channels in Figure~\ref{fig:ChannelModel} should have the larger capacity.
Our main result establishes that, for $\bar L \geq 1$, the capacity of the SSC is given by
\begin{align}
    C_{\rm SSC} = 1 - e^{-c\left(1-\frac{1}{\Bar{L}}\right)}.
\end{align}
Notice that the dependence on $\bar L$ appears as the term $(1-1/\bar{L})$ in the exponent and, 
as $c \to \infty$, $C_{\rm SSC} \to 1$ for any $\bar{L} > 1$.
This is in contrast to the shuffling-sampling channel, where $C_{\rm shuf} \to 1-1/\bar{L}$ as we increase the coverage depth $c$ to infinity.
Therefore, even in the high coverage depth regime, if $\bar{L} \approx 1$, $C_{\text{shuf}} \approx 0$. 
Furthermore, it can be verified that $C_{\rm shuf} < C_{\rm SSC}$ for any $c$ and $\bar{L}$, establishing the advantage (from a capacity standpoint) of storing data on a long molecule of DNA as opposed to on many short molecules.



The above result also allows for an interesting comparison with the uncoded setting (i.e., the genome sequencing problem) of \cite{MotahariDNA}.
When we allow coding over the string, the critical threshold on the read length reduces to $\Bar{L} > 1$, compared to  $\Bar{L} > 2$ for the uncoded setting.
Moreover, in the SSC it is possible to achieve a capacity close to $1$ by having the coverage depth be a large constant, while in the uncoded case the $c$ needs to grow as $\log n$.

%

\textbf{Background and Related Work:}
The first prototypes of DNA storage systems were presented in 2012 and 2013, when groups lead by Church~\cite{church_next-generation_2012} and Goldman~\cite{goldman_towards_2013} independently stored about a megabyte of data in DNA. 
In 2015, Grass et al.~\cite{grass_robust_2015} demonstrated that millenia long storage times are possible by protecting the data  using error-correcting codes. 
Yazdi et al~\cite{yazdi_rewritable_2015} showed how to selectively access parts of the stored data, and in 2017, Erlich and Zielinski~\cite{erlich_dna_2016} demonstrated that practical DNA storage can achieve very high information densities. In 2018, Organick et al.~\cite{organick_scaling_2017} scaled up these techniques and stored about 200 megabytes of data.
We point out that, in all of these prototypes, data is stored on many short DNA molecules, as opposed to storing it in a single very-long DNA molecule.
This is because synthesizing long strands of DNA is prohibitively expensive with current technology. Hence, this work seeks to answer what storage rates could be achieved if we were able to synthesize long DNA molecules at reasonable costs.

The prospect of practical DNA-based storage has motivated a significant amount of research into its theoretical underpinnings.
In particular, the idea of coding over a set of short strings that are then shuffled and sampled was studied in several settings~\cite{DNAStorageISIT,noisyshuffling,DNAStorageIT,LenzAnchor,lenz2018coding,lenz_upper_2019,lenz2020achieving}.
Many works have also focused on developing explicit codes tailored to specific aspects of DNA storage. These include DNA synthesis constraints such as sequence composition  \cite{kiah_codes_2016,yazdi_rewritable_2015,erlich_dna_2016}, the asymmetric nature of the DNA sequencing error channel \cite{gabrys_asymmetric_2015}, 
the need for codes that correct insertion errors \cite{sala_insertions_2016}, 
and the need for techniques to allow random access \cite{yazdi_rewritable_2015}.

The problem of reconstructing strings from a set of its subsequences has also been considered in various settings.
Several works considered studied the problem of genome sequencing and assembly from an information-theoretic standpoint \cite{MotahariDNA,BBT,shomorony2016fundamental,shomorony2016information}.
The trace reconstruction problem is another related setting where one observes (non-contiguous) subsequences of the input sequence and attempt to reconstruct it \cite{holenstein2008trace,srinivasavaradhan2018maximum,cheraghchi2019coded}.

A very relevant related setting is the problem of reconstructing a string from its substring spectrum \cite{gabrys2018unique,marcovich2019reconstruction}. Our setting is similar to this problem in two ways: (i) that both problems look at trying to reconstruct strings from substrings of fixed lengths, in general with overlaps, and (ii) the string is chosen from a codebook. However, these works have focused on the setting where a noisy substring spectrum (the multi-set of all substrings) is available, while we consider that a fixed number of reads (or substrings) are extracted from random locations.
Moreover, these works proposed explicit code constructions, while we focus on the problem of capacity characterization.






\section{Problem Setting}
\label{Section:ProblemFormulation}

We consider the Shotgun Sequencing Channel (SSC), shown in Figure~\ref{fig:ChannelModel}(a). 
The transmitter sends a length-$n$ binary string $X^n \in \{0,1\}^n$, corresponding to a message $W \in [1:2^{nR}]$. The channel output is a set of length-$L$ binary strings $\Y$. 
The channel chooses $K$ starting points  uniformly at random, represented by the random vector $T^K \in [1:n]^K$. 
The vector $T^K$ is assumed to be sorted in a non-decreasing order. 
Length-$L$ reads are then sampled with $T_i$, $i=1,\dots,K$  as their starting points. 
We allow the reads to ``wrap around'' $X^n$; i.e., if for any $i$, $T_i + L > n$, we concatenate bits from the start of $X^n$ to form length-$L$ reads. 
For example if $T_i = n - 2$ and $L =5$, then the read $\vec{Y}$ associated with this starting location is
\begin{align*}
    \vec{Y} = [X_{n-2},X_{n-1},X_n,X_1,X_2].
\end{align*}
Notice that the channel effectively treats the codeword as circular, equivalent to the circular DNA model considered in \cite{MotahariDNA}.
The unordered multi-set $\Y = \{\vec{Y}_1,\vec{Y}_2,\dots,\vec{Y}_K\}$ of reads resulting from this sampling process is the channel output.

The expected number of times a given symbol from $X^n$ is sequenced is defined as the coverage depth $c$. This is given by the expression
\begin{align*}
    c := \frac{KL}{n}.
\end{align*}

We focus on the regime where the length of the reads sampled is much smaller than the block length $n$.
In particular, as shown in previous works \cite{DNAStorageISIT,TPCglobecom,TPCLP,shomorony2021torn,nassirpour2020embedded}, the regime $L = \Theta(\log{n})$ is of interest from a capacity standpoint.
Hence, as in \cite{MotahariDNA}, we fix a normalized length $\Bar{L}$ and define
\begin{align*}
    L := \Bar{L}\log{n}.
\end{align*}
Notice that, in this regime, the total number of reads is
\begin{align*}
    K = \frac{c n}{\Bar{L}\log{n}} = \Theta\left(\frac{n}{\log{n}}\right),
\end{align*}
which is a $\log{n}$ factor smaller than what is needed in the uncoded setting from \cite{MotahariDNA}.

We define an achievable rate in the usual way.
More precisely, a $(2^{nR},n)$-code consists of a message set $[1:2^{nR}]$, an encoder that assigns codeword $x^n(W)$ to any $W \in [1:2^{nR}]$, and a decoder that assigns an estimate $\hat W(\Y) \in [1:2^{nR}]$.
A rate $R$ is achievable if there exists a sequence of $(2^{nR},n)$ codes whose error probability tends to zero as $n \to \infty$.
The capacity $C$ of the SSC is the supremum over all achievable rates $R$.


\vspace{2mm}
\noindent \textbf{Notation:} 
$\log(\cdot)$ represents the logarithm in base $2$. For functions $a(n)$ and $b(n)$, we say $a(n) = o(b(n))$ or $b(n) = \Omega(a(n))$ if $a(n)/b(n) \to 0$ as $n \to \infty$. Further, we say that a function $a(n) = \Theta(f(n))$ if there exist $n_0 \in \mathbb{N}, k_1,k_2 \in (0,\infty)$, such that  $k_1f(n) \leq a(n) \leq k_2f(n)$ $\forall n \geq n_0$. 
For an event $A$, we let $\mathbf{1}_A$ be the binary indicator of $A$. For a set $B$, $|B|$ indicates the cardinality of that set.


\section{Main Results}
\label{Section:MainResults}

The DNA storage problem considered here has two important properties: (i) the reads in general overlap with each other and (ii) the set of reads is unordered. 
Property (i) was explored in the context of genome sequencing \cite{MotahariDNA}.
Intuitively, the overlaps between the reads allow them to be merged in order to reconstruct longer substrings of $X^n$.
Property (ii) has been analyzed before in the context of several works on DNA storage. 
In particular, in the context of the shuffling-sampling channel from \cite{DNAStorageISIT}, illustrated in Figure~\ref{fig:ChannelModel}(b), 
the input to the channel is a set of strings of length $L$, and the capacity is given by $C_{\rm shuf} = (1-e^{-c})(1-1/\bar L)$.

Notice that, in the case of the shuffling-sampling channel, the output strings have no overlaps (they can only be non-overlapping or identical).
In the context of the SSC, on the other hand, the overlaps can provide useful information to fight the lack of ordering of the output strings.
Our main result captures the capacity gains that can be achieved by optimally exploiting the overlaps.
Specifically, we characterize the capacity of the SSC for any coverage depth $c$ and normalized read length $\bar L$.

\begin{theorem}\label{Thm:Main}
For any $c > 0$ and $\bar L > 0$, the capacity of the Shotgun Sequencing Channel is 
\begin{align}\label{eq:SSCCapacity}
    C = \left(1 - e^{-c\left(1 - 1/\Bar{L}\right)}\right)^{+}.
\end{align}
\end{theorem}

In order to prove Theorem~\ref{Thm:Main}, 
we consider a random coding argument and develop a careful decoding algorithm that allows for a tight analysis of the error probability.
For the converse we use a novel constrained-genie argument, which specifically tackles  property (i).

Notice that the 
capacity of the SSC given in Theorem~\ref{Thm:Main} is zero when $\Bar{L} \leq 1$. 
An intuitive reason for this is that
when $\bar L < 1$, the number of possible distinct length-$L$ sequences is just $2^{\bar L \log n} = n^{\bar L} = o(n/\log{n}) = o(K)$, and many reads must be identical. This can be used to show that the decoder cannot discern any meaningful information from $\Y$. Section~\ref{Section:Converse} discusses this further. When $\bar{L} = 1$, this same intuition doesn't hold true, but as a consequence of the continuity of $C$, we expect $C = 0$, when $\bar L = 1$. This is indeed true as seen in Section~\ref{Section:Converse}.

In order to interpret the capacity expression in  (\ref{eq:SSCCapacity}) notice that the probability that a given symbol in $X^n$ is not sequenced by any of the $K$ reads is
\al{ \label{eq:expected_coverage}
\left( 1 - \frac{L}{n} \right)^K = \left( 1 - \frac{L}{n} \right)^{\frac{cn}{L}} \to e^{-c},
}
as $n \to \infty$.
Hence the expected fraction of symbols in $X^n$ that are covered by at least one read is asymptotically close to $1-e^{-c}$.
If instead of reads of length $L = \bar L \log{n}$ we had reads of length $(\bar L -1) \log n$, the new  coverage depth would be 
\aln{
c' = \frac{K (\bar L -1) \log n}{n} = c (1 - 1/\bar{L}),
}
and the expected fraction of symbols in $X^n$ that would be sequenced would be $1-e^{-c'} = 1-e^{-c(1-1/\bar L)}$.
Hence,
the capacity expression in Theorem~\ref{Thm:Main} suggests that, on average, $\log n$ bits from each read are used for ordering information, while the remaining $(\bar L - 1) \log n$ bits provide new data information.

It is also interesting to compare the capacity of the SSC and the capacity of the shuffling-sampling channel $C_{\rm shuf} = (1-e^{-c})(1-1/\bar L)$.
%
%
%
Note from Figure~\ref{fig:CapacityComparisons}, that $C_{\rm shuf}$ is strictly upper bounded by \eqref{eq:SSCCapacity}. This shows that given a coverage depth $c$, 
there are significant gains in terms of capacity to be obtained if we store data on a long DNA molecules instead of many short DNA molecules.
Moreover if we let the coverage depth $c \to \infty$; i.e. allow for a large number of samples, when $\Bar{L} > 1$,
$C_{\rm shuf} \to 1- 1/\bar L$, while $C \to 1$.
In particular, when reads are very short and $\bar L \approx 1$, 
$C_{\text{shuf}} \approx 0$, while the capacity of the SSC can be close to $1$ for large enough $c$.

%

It is also interesting to compare Theorem~\ref{Thm:Main} to the uncoded setting of genome sequencing studied in \cite{MotahariDNA}.
As discussed in Section~\ref{Section:Introduction} and illustrated in Figure~\ref{fig:regions}, the results in \cite{MotahariDNA} show that the perfect reconstruction  reconstruction (with error asymptotically going to $0$) of a random (uncoded) string $X^n$, can be done as long as $\Bar{L} > 2$ and 
$K = \Theta(n)$.  
In contrast, for the coding setting of the SSC, as long as $\Bar{L} > 1$ and $K = \Theta\left(\frac{n}{\log{n}}\right)$, we can obtain a positive capacity. Moreover as discussed in the introduction, since as $c \to \infty$, $C = 1$ we can claim that if $K = \Omega(n/\log{n})$, then for $\Bar{L} > 1$, $C = 1$.
This means that coding allows us to considerably reduce the threshold on sampled read size (by a factor of half) and the number of samples (by a factor of nearly $1/\log{n}$), while still admitting asymptotically perfect reconstruction.

The remainder of the paper is organized as follows.
In Section~\ref{Section:AchievableRates} we prove the achievability of Theorem~\ref{Thm:Main} and in Section~\ref{Section:Converse} we prove the converse.
We conclude the paper in Section~\ref{Section:Conclusion}.

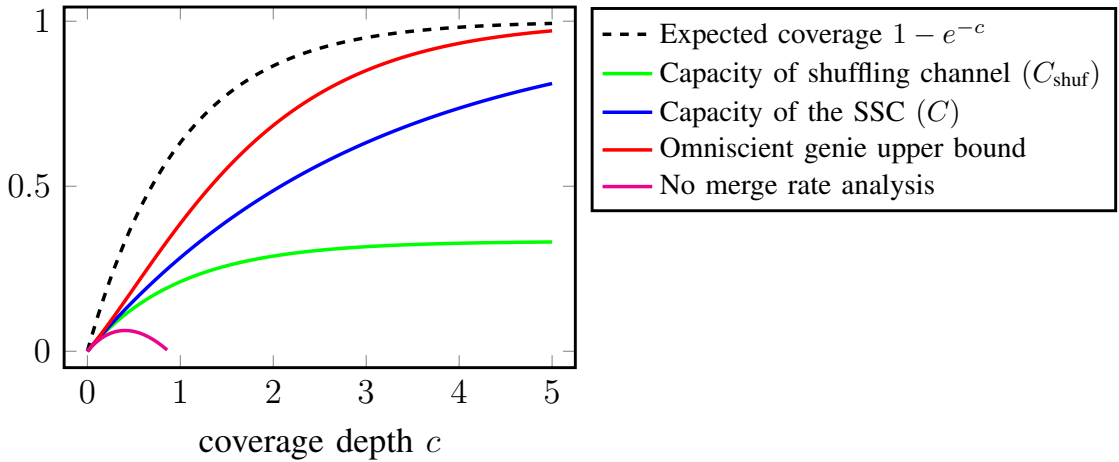
\begin{figure}[ht]
\centering 
\begin{tikzpicture}[scale=1.1]

\begin{axis}[enlargelimits=0.05,xlabel= coverage depth $c$,line width=1pt,grid=none,
legend cell align={left},legend pos=outer north east, legend style={font=\small},
width=0.47\textwidth,
height=0.36\textwidth,
]
\addplot[dashed ,very thick, name path=f,color=black,mark=none,draw=black!] table[x index=0,y index=1]{./DataFilesForGraphs/Coverage.dat};   
\addlegendentry{Expected coverage $1-e^{-c}$}

\addplot[very thick, name path=g,color=green,mark=none,draw=green!] table[x index=0,y index=1]{./DataFilesForGraphs/C_shuf.dat};   
\addlegendentry{Capacity of shuffling channel $(C_{\text{shuf}})$}

\addplot[very thick, name path=g,color=blue,mark=none,draw=blue!] table[x index=0,y index=1]{./DataFilesForGraphs/C.dat};   
\addlegendentry{Capacity of the SSC $(C)$}

\addplot[very thick, name path=g,color=red,mark=none,draw=red!] table[x index=0,y index=1]{./DataFilesForGraphs/C_full.dat};   
\addlegendentry{Omniscient genie upper bound}

\addplot[very thick, name path=g,color=magenta,mark=none,draw=magenta!] table[x index=0,y index=1]{./DataFilesForGraphs/C_no.dat};
\addlegendentry{No merge rate analysis}

\end{axis}
\end{tikzpicture}  
\vspace{2mm}
\caption{\label{fig:CapacityComparisons}
Comparison between the capacity of the SSC $C = 1 - e^{-c\left(1-\frac{1}{\bar{L}}\right)}$ with $\Bar{L} = 1.5$, the capacity of the shuffling channel with fragments of deterministic length $L$, the maximum rates achieved on the SSC when we allow a genie to merge the reads and the maximum rate discerned if reads aren't merged.}
\vspace{-2mm}
\end{figure} 

\section{Achievability}
\label{Section:AchievableRates}

We use a random coding argument to prove the achievability of Theorem~\ref{Thm:Main}. We generate a codebook 
with $2^{nR}$ codewords of length $n$,
independently picking each letter $\text{Ber}(1/2)$. Let the codebook be $\mathcal{C} = \{\mathbf{x}_1,\mathbf{x}_2,\dots, \mathbf{x}_{2^{nR}}\}$. The encoder chooses the codeword corresponding to the message $W \in [1:2^{nR}]$, and sends $\mathbf{x}_{W}$ across the Shotgun Sequencing Channel. The output, $\Y$, is presented to the decoder. For the analysis and without loss of generality, we assume  $W = 1$.

The optimal decoder looks for a codeword that contains all the reads in $\Y$ as substrings. 
Analyzing the error probability of this optimal decoder, however, is hard. 
We therefore aim to develop a decoding rule that is simple enough to analyze.

\subsection{Analysis without exploiting overlaps}

The fact that, in general, there are overlaps between the reads is an important feature of the channel output, since they allow reads to be merged, and this should be taken into account while developing a decoding rule. To motivate this, let us first bound the error probability without exploiting the overlaps for merging reads.


We say that the $i$th bit of $X^n$ is \emph{covered} if there is a read with starting position in 
\aln{
\{i-L+1,i-L+2,\dots,i\},
}
where the indices wrap around $X^n$. 
We then define the \emph{coverage} as the random variable
\al{ \label{eq:coverage}
    \Phi = \frac{1}{n}\sum_{i=1}^{n}\mathbf{1}_{\{i\text{th bit is covered}\}}.
}

\begin{restatable}{lemma}{CoverageLemma}
\label{lem:coverage}(Coverage)
For any $\epsilon > 0$, the coverage $\Phi$ satisfies
\begin{align}\label{eqn:ConcBoundCoverage}
     \Pr\left(\left| \Phi - \left(1 - e^{- c}\right)\right| > \epsilon\left(1 - e^{- c}\right)\right) \to 0,
\end{align} 
as $n \to \infty$.
\end{restatable}
The proof of this lemma is presented in Appendix~\ref{CovLemma}.
Note that $\lim_{n \to \infty} E[\Phi] = 1-e^{-c}$ as described in (\ref{eq:expected_coverage}), and thus, Lemma~\ref{lem:coverage} guarantees that the coverage $\Phi$ is concentrated around its expected value. We use this fact to discern how many bits in a candidate codeword need to match the bits sampled in the reads.

The decoding rule we consider is as follows: The decoder looks for the codeword in the codebook that contains all reads as substrings of that codeword and that for an $\ep > 0$, the coverage of these reads $> (1-\ep)(1-e^{-c})$. It declares an error if more than one such codeword exists. We want to bound the probability of error $\mathcal{E}$ based on Lemma~\ref{lem:coverage}, 
we can define $B := (1-\epsilon)(1-e^{-c})$ and follow steps similar to \cite{TPCLP} to obtain
\begin{align}\label{eq:RateTestOne}
    \Pr(\mathcal{E}) &= \Pr(\mathcal{E}|W = 1) \leq \Pr(\mathcal{E}|W = 1, \Phi \geq B) + \Pr(\Phi < B) \nonumber \\
    &\stackrel{(a)}{\leq} 2^{nR} \times n^{K}\times\frac{1}{2^{nB}} + o(1)
    = 2^{n(R - (K \log n)/n - B)} + o(1),
\end{align}
where $(a)$ holds because there are at most $n^K$ ways to arrange the $K$ reads on a codeword and,
given an arrangement, at least $n B$ bits of an incorrect codeword would need to match our reads to create an error.
Since $(K \log n)/n = c/\bar{L}$, in order for $\Pr(\mathcal{E}) \to 0$ as $n \to \infty$, we would need 
\begin{align}
    R \leq 1 - e^{-c} - \frac{c}{\Bar{L}}.
\end{align}
The achievable rate obtained from this analysis is plotted in Figure~\ref{fig:CapacityComparisons} 
in magenta. This rate is suboptimal and in fact (above a critical value of $c$) reduces as the coverage depth increases. 
This arises because when we bound the number of ways to arrange the reads on a length-$n$ codeword by $n^K$, it does not take into account overlaps between the reads.
To be able to discern higher rates, we need to develop a way to utilize the fact that, in general, many of the reads overlap with each other.



\subsection{Using overlaps to merge reads}

The analysis above indicates the need to merge the reads before we compare them to candidate codewords. 
Unfortunately, merging reads is not a straightforward process because reads $\vec Y_i$ and $\vec Y_j$ may have an overlap even if they do not correspond to overlapping segments of $X^n$.
In general, the merging process will be prone to errors and we need to develop a decoding algorithm that 
considers merges in a careful way.

In Section~\ref{Section:ProblemFormulation}, we defined the unknown vector $T^K$ to be the ordered starting positions of the reads in $\Y$. Thus, without loss of generality we assume that $\vec{Y}_i$ starts at $T_i$. We define the \emph{successor} of $\vec Y_{i}$ as $\vec Y_{i+1}$. We assume $Y_1$ is the successor of $Y_K$.
Now we need a consistent definition to characterize how large the overlap of a given read is.


\begin{defn} \label{def:overlap}
(Overlap size) The overlap size of a read is defined as the number of bits the suffix of the read shares with its successor. It has an overlap size of $0$ if no bits are shared (i.e., if a read and its successor have no overlap).
\end{defn}

The above definition implies that the overlap size of $\vec Y_i$ is $(L-(T_{i+1} - T_i))^+$. Notice that some reads might share some of their prefix bits with a predecessor read, but we do not consider this a contribution to the overlap size of that read. Intuitively speaking, since each bit of the string $X^n$ was generated independently as a $\text{Ber}(1/2)$ random variable, we would expect larger overlap sizes to be easily discerned as compared to smaller ones. Therefore, we would need to know: (a) how many pieces exist of particular overlap sizes, and (b) given an overlap size, how ``easy" it is to merge a read with its successor.

To handle (a), we define $G(\gamma)$ as a random variable that counts the number of reads with an overlap size of $\gamma\log{n}$, where $\gamma \in \Gamma := \left\{\frac{1}{\log{n}},\frac{2}{\log{n}},\dots,\Bar{L}\right\}$. 
Thus, $\gamma$ is chosen from a finite set that depends on $n$. 
We can say
\begin{align*}
    G(\gamma) := \sum_{i=1}^KG(\gamma)_i,
\end{align*}
where $G(\gamma)_i = \mathbf{1}_{\{\vec{Y}_i\text{ has an overlap size of }\gamma\log{n}\}}$.

To capture (b), given a binary string $\vec{z}$, we define the random variable $M_{\vec{z}}$ as the number of times $\vec{z}$ appears as the prefix of a read in $\Y$. Note that the length of $\vec{z}$ is in $ [1:\Bar{L}\log{n}]$. Let $\mathcal{Z}$ be the set of all binary strings with lengths in $[1:\Bar{L}\log{n}]$.

If we can identify the merges correctly, we are left with a set of variable-length strings called islands, which we formally define next. 
\begin{defn}
(Islands) The set of non-overlapping substrings that are obtained after merging all the reads to their successors based on positive overlap sizes are called islands.
\end{defn}

Let $K^\prime$ be the number of islands. Then, we have the following result.
\begin{restatable}{lemma}{IslandNumbers}\label{lemma:IslandNumbers}
(Number of islands) For any $\epsilon > 0$, the number of islands $K'$ satisfies
\begin{align}\label{eqn:IslandNumbers}
    \Pr\left(\middle|K' - Ke^{-c}\middle|\geq \epsilon Ke^{- c}\right)
    \to 0,
\end{align}
as $n \to \infty$.
\end{restatable}
The proof of this lemma is available in Appendix~\ref{IslandNumberLemma}.
%
Similar to the previous lemma, Lemma~\ref{lemma:IslandNumbers} guarantees that the number of islands $K^\prime$ is concentrated around its expectation $Ke^{-c}$.
Lemmas~\ref{lem:coverage} and \ref{lemma:IslandNumbers} are used in the later part of the decoding to look at different arrangements of the non-overlapping islands and the bits that match these arrangements. However, to use these results, the decoder would first need to obtain the non-overlapping islands (the decoder only has the reads currently). The following lemmas give the decoder some guidelines on how to construct these islands from the reads.
\begin{restatable}{lemma}{MergeCandidates}\label{lemma:MergeCandidates}
(Number of potential overlaps) For any $\epsilon > 0$,
\begin{align}
    &\Pr\left(\bigcup_{\vec z 
    \in \mathcal{Z}:\gamma(\vec{z})\leq 1-\ep}\left\{\left|M_{\vec{z}}-K n^{-\gamma(\vec z)}\right|\geq  \ep K n^{-\gamma(\vec z)}\right\}\right) \to 0 \text{ and }\nonumber \\
    &\Pr\left(\bigcup_{\vec z 
    \in \mathcal{Z}:\gamma(\vec{z})> 1-\ep}\left\{M_{\vec{z}} \geq n^{\ep}\right\}\right) \to 0,
\end{align}
as $n \to \infty$, where we define $\gamma(\vec{z}) := |\vec z|/\log{n}$. 
\end{restatable}

Lemma~\ref{lemma:MergeCandidates} considers two separate cases for binary strings based on their length.
For strings $\vec z$ with length at most $(1-\ep)\log{n}$, Lemma~\ref{lemma:MergeCandidates} states that $M_{\vec z}$ is close to its mean 
\aln{
K n^{-\gamma(\vec z)} = \frac{c n^{1-\gamma(\vec z)}}{\bar L \log n}.
}
For strings $\vec z$ with length greater than $(1-\ep)\log{n}$, the same concentration result does not hold, and Lemma~\ref{lemma:MergeCandidates} simply states that $M_{\vec z} < n^{\ep}$ with high probability.

\begin{restatable}{lemma}{OverlapNumber}\label{lem:OverlapNumber}
(Number of reads of a given overlap size) For all $\epsilon > 0$, 
\begin{align}
    \Pr\left(\bigcup_{\gamma \in \Gamma} \left\{ |G(\gamma) - \bar{G}(\gamma)|\geq \epsilon \bar{G}(\gamma) \right\} \right) \to 0,
\end{align}
as $n \to \infty$, where $\bar{G}(\gamma) := E[G(\gamma)]$.
\end{restatable}
Lemma~\ref{lem:OverlapNumber} give us a handle on the expected number of overlaps of each size, which will be used by the decoder when trying to construct the islands from the reads.
Lemmas~\ref{lemma:MergeCandidates} and \ref{lem:OverlapNumber} are proved in Appendices~\ref{MergeCandidatesLemma} and \ref{OverlapNumber}.

The decoding procedure starts with a brute-force search over ways to merge the reads into islands, which we refer to as the  Partition and Merge (PM) algorithm.
We will first explain it in words and follow it by outlining the exact algorithm. 
First, the decoder considers all possible partitions of the reads into $L$ groups, by assigning potential overlap sizes to each read. This can be done by looking at all ways of assigning a number in $[0:L]$ to each of the reads. To make this precise, we can look at all possible vectors of the form $\vec{p}:=(p_1,p_2,\dots,p_K) \in [0:L]^K$ and call them \emph{partition vectors}. Each element $p_i$ of the vector corresponds to an assigned overlap size of read $\vec Y_{\sigma(i)}$ for some permutation $\sigma$ of the elements of $\Y$. Thus, each partition vector along with a permutation $\sigma$ can be viewed as assigning an overlap size to each read. 
It is easy to see the total number of such partition vectors (and hence the total possible partitions) will be $P := (L+1)^K$.

Rather than considering all $P$ partitions, we will only consider partitions that satisfy the bounds implied by Lemmas~\ref{lem:coverage}--\ref{lem:OverlapNumber}. To make this requirement precise, we define for a partition vector $\vec{p}$, $G(\vec{p},\gamma)$ to be the number of reads in $\Y$ that would have an overlap size of $\gamma\log{n}$ according to partition vector $\vec p$, which can be written as 
\begin{align}
    G(\vec{p},\gamma) = |\{ i : p_i = \gamma \log n \}|.
\end{align}
Note that since the number of potential islands is exactly equal to the number of reads with overlap size zero, the total number of islands according to $\vec{p}$ is $G(\vec{p},0)$.
Moreover we define $\Phi(\vec{p})$ as the total coverage of the reads according to $\vec{p}$, which is given by
\begin{align}
    \Phi(\vec{p}) := KL - \sum_{i=1}^K p_i.
\end{align}
We then define the set $\P$ as the set of all $\vec{p}$  such that (for a fixed $\ep > 0$):
\begin{itemize}
    \item $|\Phi(\vec{p}) - (1-e^{-c})| \leq \ep (1-e^{-c})$ (i.e., coverage is close to expected coverage),
    \item $|G(\vec{p},0) - Ke^{-c}|\leq \ep Ke^{-c}$ (i.e., number of islands is close to expected number of islands),
    \item $|G(\vec{p},\gamma) - \bar G(\gamma)| \leq \ep \bar G(\gamma)$ for all $\gamma \in \Gamma$ (i.e., number of reads with overlap size $\gamma \log n$ is close to the expected number).
\end{itemize}
Therefore, $\P$ restricts the total number of partition vectors to a smaller set of partition vectors that are admissible according to Lemmas~\ref{lem:coverage},~\ref{lemma:IslandNumbers} and \ref{lem:OverlapNumber}.
 
Now, for each partition vector $\vec{p} \in \P$, we take all possible $K!$ permutations $\sigma$ of the reads. For each permutation, all of the reads are compared to their successors.
If every read can be successfully merged with its successor with the assigned overlap size, we retain the set of substrings formed after these merges as a Candidate Island set, and add it to the set ${\rm CI}$.
Notice that ${\rm CI}$ is a set of sets of variable-length strings.
This procedure is summarized in Figure~\ref{fig:DecodingRule} 
and Algorithm~\ref{alg:DecodingRule}.
After completion of Algorithm~\ref{alg:DecodingRule}, the decoder checks, for each set of candidate islands in ${\rm CI}$, whether there exists a codeword that contains all the candidate islands as substrings.
If only one such codeword is found, the decoder outputs its index.
Otherwise, an error is declared.


\SetAlgoNoLine%
  \begin{algorithm}
  \DontPrintSemicolon
    \For{each partition vector $\vec p \in \P$} {
      \For{each permutation $\sigma$ of $[1:K]$} {
        check if suffix of length $p_i$ of $\vec Y_{\sigma(i)}$ matches prefix of $\vec Y_{\sigma(i+1)}$, for $i=1,\dots,K$ \;        
        \If{prefix and suffix match for $i=1,\dots,K$} {
            Merge reads according to overlaps \;
            Add set of resulting islands to ${\rm CI}$ \; 
        } 
      }
    }
    \Return{${\rm CI}$}\;
    \caption{Partition and Merge\label{alg:DecodingRule}}
  \end{algorithm}

\begin{figure}
\vspace{-4mm}
\centering
\includegraphics[width=0.9\linewidth]{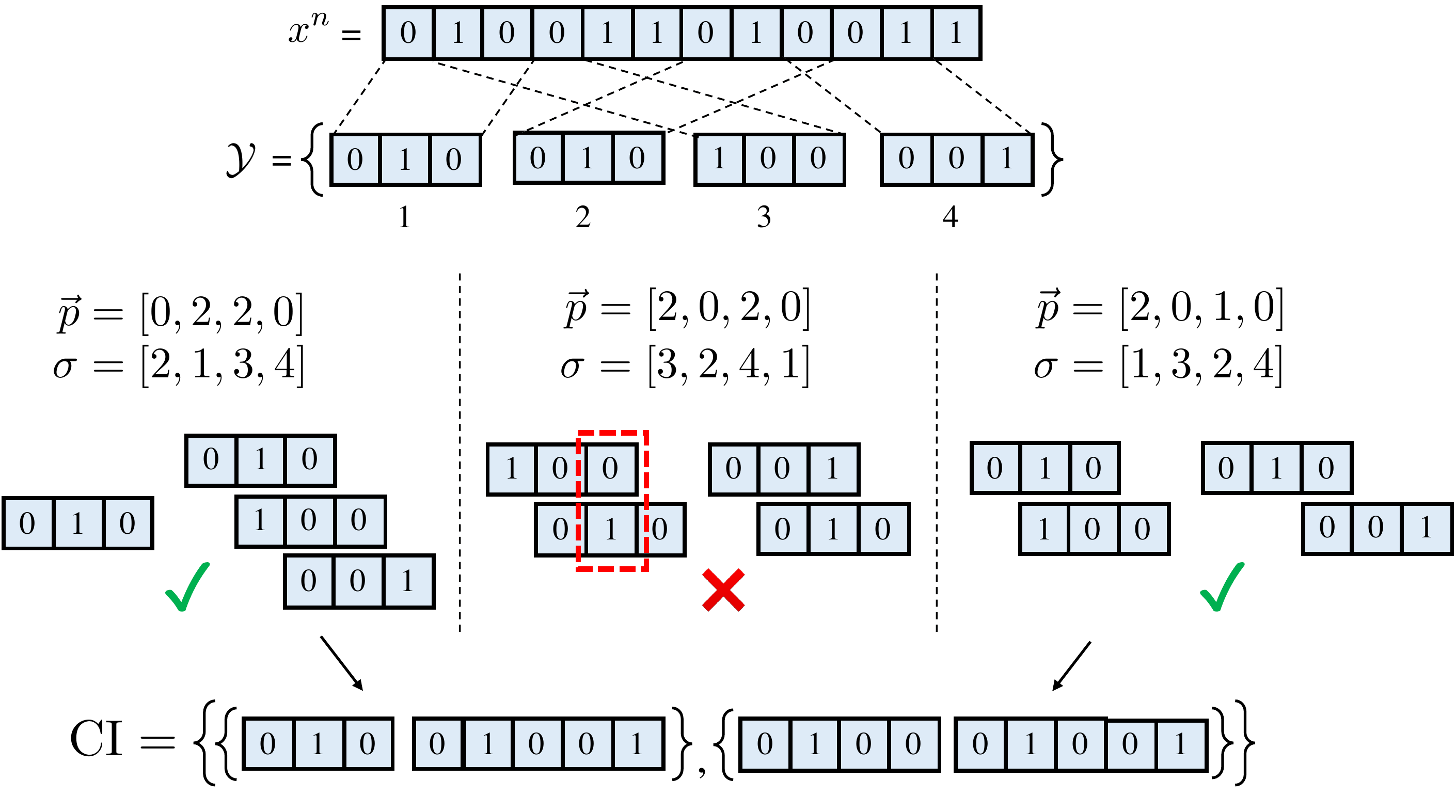}
\caption{The decoder receives the shotgun sequenced reads $\Y$ and performs the Partition and Merge procedure.
For each partition $\vec p \in [0:L]^K$ of the $K$ reads according to overlap size, and each ordering of the reads $\sigma$, the decoder attempts to merge the reads into islands based on $\sigma$ and $\vec p$.
The figure shows this procedure for three choices of $\vec p$ (out of $(1+L)^K$) and three choices of $\sigma$ (out of $K!$).
If the merging of all reads is successful for some $p$ and $\sigma$, the set of resulting islands is added to the set ${\rm CI}$.
}
\label{fig:DecodingRule}

\end{figure}




Let the event that the decoder makes an error be $\mathcal{E}$. An error occurs if more than one codeword contains any of the CI sets as substrings.
We define $B_1 := (1+\epsilon)Ke^{- c}$, $B_2 := (1-\epsilon)\left(1-e^{- c}\right)$, $B_{3}(\gamma) := (1+\epsilon)n^{1-\gamma}$ for $\gamma \leq 1 - \ep$, $B_{3}(\gamma) = n^{\ep}$ for $\gamma > 1-\ep$ and $B_{4}(\gamma) := (1+\epsilon)\bar{G}(\gamma)$,
and we define the corresponding undesired events as
\aln{
& \B_1 = \{K' > B_1\} \\
& \B_2 = \{\Phi < B_2\} \\
& \B_3 = \bigcup_{\vec{z} \in \mathcal{Z}} \{ M_{\vec{z}} > B_{3}(\gamma(\vec z)) \}\\
& \B_4 = \bigcup_{\gamma \in \Gamma} \{  G(\gamma) > B_{4}(\gamma) \}.
}
From Lemmas~\ref{lem:coverage},~\ref{lemma:IslandNumbers},~\ref{lemma:MergeCandidates} and~\ref{lem:OverlapNumber}, if we let 
\begin{align}\label{eq:Bdefn}
\B = \B_1 \cup \B_2 \cup \B_3 \cup \B_4,
\end{align}
we have $\Pr(\B) \to 0$.
Note that conditioned on $\bar{\B}$, we are guaranteed that the ${\rm CI}$ set outputs the true island set. This is because exactly one partition and one arrangement given that partition correspond to the true order in which the reads were sampled.
%
 Before we use this to bound the probability of error, notice that the error event depends on the total number of CI sets output by the PM algorithm. In general, this is not a deterministic value. 
 We will define $\overline{CI}_n$ as an upper bound on the number of CI sets conditioned on $\overline \B$. 
 We claim that
conditioned on $\overline \B$, after the PM algorithm, the resulting ${\rm CI}$ (which is a set of sets of binary strings) satisfies
\al{ \label{eq:cibound}
|{\rm CI}| \leq P \times \prod_{\gamma \leq 1 - \ep} B_{3}(\gamma)^{B_4(\gamma)}\times\prod_{\gamma > 1-\ep} n^{\ep B_4(\gamma)} := \overline{CI}_n.
}
To see this, first we notice that $|\P| \leq P$.
According to a given partition vector $\vec p \in \P$, there are at most $B_4(\gamma)$  reads with overlap size $\gamma\log{n}$.
Given a read $\vec{Y}_i$ with assigned overlap size $\gamma \log n$, when $\gamma \leq 1-\ep$,
there are at most $B_3(\gamma)$ reads whose prefix matches the $(\gamma \log n)$-suffix of $\vec{Y}_i$ and, therefore, at most $B_3(\gamma)$ potential valid merges.
Therefore, for a given overlap size $\gamma \log{n}$, $\gamma \leq 1$, there at most $B_{3}(\gamma)^{B_{4}(\gamma)}$ merge possibilities. 
However, when $\gamma > 1-\ep$, we know that for the given read, there at most $n^{\ep}$ potential valid merges. Therefore there are at most $n^{\ep B_4{(\gamma)}}$ merge possibilities. 
We thus bound the probability of error averaged over all codebooks as
\al{
    \Pr(\E) & = \Pr(\mathcal{E}|W=1) \leq \Pr(\mathcal{E}|W=1,\overline{\mathcal{B}}) + \Pr(\mathcal{B}) \nonumber \\
    &\stackrel{(a)}{\leq}  2^{nR} \times \overline{CI}_n \times n^{B_1} \times \frac{1}{2^{nB_2}} + o(1) \nonumber \\ 
    &=  2^{nR + \log \overline{CI}_n + B_1\log{n} - nB_2} + o(1) \nonumber \\
    &= 2^{nR + \log \overline{CI}_n + (1+\ep)Ke^{-c}\log{n} - n(1-\ep)(1-e^{-c})} + o(1).\label{eq:error_part1}
}

This follows because an error occurs if any of the $2^{nR}-1$ codewords ($W \ne 1$) contain any of the sets of candidate islands in ${\rm CI}$ (which is upper bounded by $\overline{CI}_n$ when conditioned on $\overline{\B}$).
Each of the sets in ${\rm CI}$ contains at most $B_1$ islands and a total island length of at least $nB_2$ bits.
Hence, there are at most $n^{B_1}$ ways to arrange the islands on a codeword and,
given one such arrangement, an
an erroneous codeword must match these islands at at least $n B_2$ bits.

Let us compare \eqref{eq:RateTestOne} and \eqref{eq:error_part1}. The term $n^K$ in \eqref{eq:RateTestOne}, which bounds the number of arrangements of reads on a candidate codeword, is replaced by $n^{B_1} \times \overline{CI}_n$ in \eqref{eq:error_part1}, which carefully takes into account both the cost of merging reads together and the number of arrangements after all the reads are merged together to form CIs. This improvement, as we see below, is crucial to achieve optimal rates.

In order to have $\Pr(\mathcal{E}) \to 0$ in \eqref{eq:error_part1}, we require
\begin{align*}
    R &\leq\lim_{n \to \infty}\left( (1-\epsilon)(1-e^{-c}) - (1+\epsilon)\frac{ce^{-c}}{\Bar{L}} - \frac1n\log{\overline{CI}_n}\right) \\
    &= (1-\epsilon)(1-e^{-c}) - (1+\epsilon)\frac{ce^{-c}}{\Bar{L}} - \lim_{n \to \infty} \frac1n\log{\overline{CI}_n}
\end{align*}
The following lemma, proved in Appendix~\ref{CostSumLem}, evaluates the last term in the above expression.
\begin{restatable}{lemma}{CostSumEvaluation}\label{lem:CostSumEvaluation}
(Cost of merging)
The upper bound $\overline{CI}_n$ on $|{\rm CI}|$ satisfies
\aln{
\lim_{n \to \infty} \frac{1}{n}\log \overline{CI}_n \leq e^{-c\left(1-\frac{1}{\Bar{L}}\right)}-\left(\frac{c}{\Bar{L}}+1\right)e^{-c} + f(\ep),
}
where $f(\ep) \to 0$ as $\ep \to 0$.
\end{restatable}
This Lemma is proved in Appendix~\ref{CostSumLem}. From Lemma~\ref{lem:CostSumEvaluation} and by letting $\epsilon \to 0^+$, we can conclude that all rates
\begin{align*}
    R < 1 - e^{-c\left(1 - \frac{1}{\Bar{L}}\right)},
\end{align*}
are achievable. This completes the achievability proof of Theorem~\ref{Thm:Main}.

\section{Converse}
\label{Section:Converse}

To prove the converse, we borrow some insights from the achievable scheme in Section~\ref{Section:AchievableRates}. The idea is to have a genie-aided channel in which the reads are already merged into islands. However, this step needs to be carried out with care. More specifically, an omniscient genie, which merges all correct overlaps, would be too powerful and therefore result in a loose upper bound. Instead, we use a constrained genie that can only merge reads with overlap sizes above a certain threshold, and the resulting upper bound matches achievable rates of Section~\ref{Section:AchievableRates}.

The proof is organized as follows. We first look at the case when $\Bar{L} \leq 1$ and show the capacity is zero. Next, we provide the argument with an omniscient genie, which as we mentioned, results in a loose upper bound. We then show how to carefully constrain this genie to obtain a tighter upper bound that matches the achievable rates.


 \underline{\textbf{Short reads ($\Bar{L} < 1$).}} 
The intuition is that, when $\bar L < 1$, the number of possible distinct length-$L$ sequences is just $2^{\bar L \log n} = n^{\bar L} = o(n/\log{n}) = o(K)$, and many reads must be identical.
%
%

By Fano's inequality, followed by the counting argument similar to \cite{DNAStorageISIT}, we have that
\begin{align*}
    R &\leq  \frac1n H(\Y) \stackrel{(a)}{\leq} \frac1n \log\binom{K + 2^L -1}{K} \leq \frac Kn \log\left(\frac{K+2^L-1}{K}\right) \\
    &= \frac{K}{n}\log\left(1+\left(\frac{2^L - 1}{K}\right)\right) \stackrel{(b)}{\leq} K\frac{2^L}{nK} = cn^{\Bar{L} - 1} \to 0, 
\end{align*}
as $n \to \infty$, when $\Bar{L} < 1$. In the above set of equations $(a)$ is due to Lemma~1 in \cite{DNAStorageISIT} and $(b)$ is because $\log(1+x)\leq x,$ when $x > 0$.  Note that this holds irrespective of the value of $K$.

\underline{\textbf{Omniscient genie.}} As observed in Section~\ref{Section:AchievableRates}, merging the reads correctly to form islands is helpful while decoding the message. To capture this in the converse, suppose we have a genie that merges the reads into islands a priori and thus, forms variable-length non-overlapping substrings. Let $\Y^\prime$ be the multiset of the islands this genie creates from $\Y$.
Then, from Fano's inequality,
\begin{align}\label{eq:FanoStep}
    R &\leq \lim_{n \to \infty} \frac{1}{n} I(X^n;\Y) \leq \lim_{n \to \infty} \frac{H(\Y)}{n} 
    \nonumber \\
    &\leq \lim_{n \to \infty} \frac{H(\Y,\Y^\prime)}{n} =  \lim_{n \to \infty} \frac{H(\Y^\prime) + H(\Y|\Y^\prime)}{n}.
\end{align}

Consider the second term in the numerator of the above expression. Notice that this term intuitively looks at the uncertainty in the set of reads given the fully merged islands. It would be helpful to get a handle on the number of reads per island to bound this term. Thus we define $D$ as the maximum number of reads making up an island, and we have the following lemma.
\begin{restatable}{lemma}{MaxIslandsBound}\label{lem:MaxIslandsBound} 
(Maximum number of reads per island)
For any $\gamma_0 > -1/\log \left(1-e^{-c}\right)$, as $n \to \infty$,
\begin{align}\label{eqn:ConcBoundReordering}
    \Pr\left(D > \gamma_0 \log{n}\right)
    \to 0.
\end{align}
\end{restatable}
This lemma is proved in Appendix~\ref{MaxIslandBound}. 
Now define $\mathcal{B}$ as in \eqref{eq:Bdefn}. 
For a fixed $\gamma_0 > -1/\log \left(1-e^{-c}\right)$,
\begin{align}\label{eq:TheIslandTerm}
    H(\Y|\Y^\prime) &\leq H(\Y,\1_{\bar{\B},\, D \leq \gamma_0 \log n} | \Y^{\prime})
\leq 1 + H(\Y | \Y^{\prime},\1_{\bar{\B},\, D \leq \gamma_0 \log n}) \nonumber\\
    & \leq 1 +H(\Y|\Y^\prime,\overline{\mathcal{B}},D\leq \gamma_0\log{n})\Pr(\overline{\mathcal{B}},D\leq \gamma_0\log{n}) \nonumber \\
    & \quad + H(\Y|\Y^\prime,\{\mathcal{B}\text{ or }D > \gamma_0\log{n}\}))\left(\Pr(\mathcal{B} \text{ or } D> \gamma_0\log{n})\right) \nonumber \\
    &\leq 1 + H(\Y|\Y^\prime,\overline{\mathcal{B}},D\leq \gamma_0\log{n}) \nonumber \\
    & \quad + H(\Y|\{\mathcal{B}\text{ or }D > \gamma_0\log{n}\})(\Pr(\mathcal{B}) + \Pr(D> \gamma_0\log{n})).
\end{align}




Now, since $\Y$ is fully determined by $X^n$ and the read starting points $T^K$, we have that $H(\Y|\{\mathcal{B}\text{ or }D > \gamma_0\log{n}\}) \leq 2n$. We can thus claim that
\begin{align}\label{eq:SecondTermLim}
    &\lim_{n \to \infty} \frac1nH(\Y|\{\mathcal{B}\text{ or }D > \gamma_0\log{n}\})(\Pr(\mathcal{B}) + \Pr(D> \gamma_0\log{n})) \\
    &\quad \quad \stackrel{(a)}{\leq} \lim_{n \to \infty} 2(\Pr(\mathcal{B}) + \Pr(D> \gamma_0\log{n})) = 0,
\end{align}
where $(a)$ is due to the fact that $\Pr(\mathcal{B}) \to 0$ as $n \to \infty$ and Lemma~\ref{lem:MaxIslandsBound}.

We now focus on the first entropy term in \eqref{eq:TheIslandTerm}. Let $Q$ be the total number of substrings of length-$L$ in an island. Given the maximum number of reads per island $D \leq \gamma_0\log{n}$, we can say that the total length of an island cannot exceed $L\times\gamma_0\log{n} = \bar{L}\gamma_0\log^2{n}$. Therefore the total number of substrings per island $Q$ cannot exceed
\begin{align}
    Q \leq \bar{L}\gamma_0\log^2{n}.
\end{align}
Since there are $K^{\prime} \leq K$ islands in $\Y^{\prime}$, we can say that the total number of substrings of length-$L$ in $\Y^{\prime}$ is upper bounded by $KQ \leq K\times\bar{L}\gamma_0\log^2{n} = \gamma_0 cn\log{n}$. 

Now for the term $H(\Y|\Y^{\prime},\overline{\B},D\leq\gamma_0\log{n})$, $\Y$ can be thought of as an histogram over these $KQ$ length-$L$ substrings with the sum of the histogram entries being exactly $K$.
%
%
 Following the counting argument from Lemma~1 in \cite{DNAStorageISIT}, we can say
\begin{align}
    &H(\Y|\Y^{\prime},\overline{\mathcal{B}},D\leq \gamma_0\log{n}) \leq \log\binom{KQ + K -1}{K} \nonumber \\
    &\leq K\log\left(\frac{e(KQ + K - 1)}{K}\right)
    \leq \frac{cn}{\log{n}}\log{\left((ec\Bar{L}\gamma_0)\log^2{n}\right)}.
\end{align}
This implies that
\begin{align}\label{eq:FirstTermLim}
    \lim_{n \to \infty} \frac1n(1+H(\Y|\Y^{\prime},\overline{\mathcal{B}},&D\leq \gamma_0\log{n})) = 0.
\end{align}
Therefore from equations \eqref{eq:FirstTermLim} and \eqref{eq:SecondTermLim}, \eqref{eq:FanoStep} becomes
\begin{align}\label{eq:FanoReduced}
    R \leq \lim_{n \to \infty} \frac1n H(\Y^{\prime}).
\end{align}
As before, let $K'$ be the number of elements in $\Y^{\prime}$ and let the random variable $N_1,N_2,\dots,N_{K'}$ be the lengths of the islands. The result in \cite{TPCLP} gives us a way to upper bound this entropy term. It showed that the entropy of unordered sets of variable length binary strings can be upper bounded by the difference of two terms as ``cumulative coverage depth $-$ reordering cost", two terms that were introduced in \cite{TPCLP}.
Precisely it showed that for an unordered set $\Y^{\prime}$ with binary strings of lengths given by the random variables $N_1,N_2,\dots,N_{K'}$, where $K'$ is also a random variable we can say that
\begin{align*}
    \lim_{n \to \infty}  \frac{1}{n}H(\Y^{\prime}) &\leq \lim_{n \to \infty}\left(\frac{1}{n}E[K']E[N_1] - E[K']\frac{\log{n}}{n} \right) \\
    &\stackrel{(a)}={} (1-e^{-c}) - \frac{c}{\Bar{L}}e^{-c},
\end{align*}
when $\lim_{n\to \infty} \frac{\log{n}}{E[N_i]} \in (0,\infty)$ and $E[N_1^2/(\log{n})^2]$ is finite and bounded. This is indeed true and is proved in Appendix~\ref{FiniteIslandLength}. Here $(a)$ is due to Lemmas~\ref{lem:coverage} and \ref{lemma:IslandNumbers}.

This is an upper bound to the achievable rate we seek to match from Section~\ref{Section:AchievableRates}. Figure~\ref{fig:CapacityComparisons} shows that this is in fact a strict upper bound.

This intuitively indicates that the genie we use is probably too powerful. This makes sense, since discerning smaller overlaps probably adds a cost. In fact, our achievable rate calculated this cost (see Lemma~\ref{lem:CostSumEvaluation}).

\underline{\textbf{Constrained genie.}} Let's re-introduce a genie, except this genie can only merge reads with overlap sizes of at least $\delta\log{n}$. Following the nomenclature in \cite{LanderWaterman}, we call these \emph{apparent islands}. Let $\Y^\prime_{\delta}$ be the set formed by this genie. Following the same steps as above, we can say that
\begin{align}\label{eq:ApparantIslandsBound}
    R \leq \lim_{n \to \infty} \frac{1}{n}H(\Y^{\prime}_{\delta}).
\end{align}

Let $K^{\prime\prime}$ be the number of apparent islands formed. Let $N^{\delta}_1,\dots,N^{\delta}_{K^{\prime\prime}}$ be the length of the islands. 
Like before we use the result in \cite{DNAStorageISIT} to bound \eqref{eq:ApparantIslandsBound}. We omit the many steps that are exactly the same as that for the omniscient genie. 
Employing the result in \cite{TPCLP} again, we obtain
\begin{align}\label{eq:ConstrainedGenieBound}
    \lim_{n \to \infty}  &\frac{1}{n}H(\Y^{\prime}_{\delta}) \leq \lim_{n \to \infty}\left(\frac{1}{n}E[K^{\prime\prime}]E[N_1] - E[K
    ^{\prime\prime}]\frac{\log{n}}{n} \right).
\end{align}
Setting $\sigma := 1-\frac{\delta}{\Bar{L}}$, it can be shown that
\begin{align}\label{eq:ApparantIslandLength}
    \lim_{n \to \infty}\left(\frac{1}{n}E[K^{\prime\prime}]E[N_1] - E[K
    ^{\prime\prime}]\frac{\log{n}}{n} \right) = (1-e^{-c\sigma}) + c(1-\sigma)e^{-c\sigma} - \frac{c}{\Bar{L}}e^{-c\sigma}.
\end{align}
The proof of this is presented in 
Appendix~\ref{ApparantIslandLengthAppendix}.
Therefore, we can say that
\begin{align}
\label{eq:minimize}
   \lim_{n \to \infty} \frac1n H(\Y^{\prime}_{\delta}) \leq 
(1-e^{-c\sigma}) + c(1-\sigma)e^{-c\sigma} - \frac{c}{\Bar{L}}e^{-c\sigma}.
\end{align}

The final step is to minimize this bound over $\sigma$. The minimum value of \eqref{eq:minimize} is achieved when $\sigma = 1 - \frac{1}{\Bar{L}}$ (or $\delta =1$). Therefore, substituting this value, we obtain
\begin{align}
    \lim_{n \to \infty} \frac1n H(\Y^{\prime}) \leq 1 - e^{-c\left(1-\frac{1}{\Bar{L}}\right)}.
\end{align}

An interesting observation here is that $\delta = 1$ implies the constrained genie should merge all reads with overlap sizes greater than $\log{n}$ for the converse to match the achievable rate. This is consistent with Lemma~\ref{lemma:MergeCandidates}, which implies that multiple merge candidates are likely to exist when the overlap size is less than or equal to $\log{n}$.
This completes the converse proof of Theorem~\ref{Thm:Main}.


\section{Conclusion}

\label{Section:Conclusion}

In this work, motivated by applications in DNA data storage, we introduced the Shotgun Sequencing Channel (SSC).
We characterized the SSC capacity exactly. This capacity was shown to strictly upper bound the capacity of a shuffling-sampling \cite{DNAStorageISIT} channel, which modelled DNA storage systems that sampled uniformly at random from a set of short length strings. In fact, we showed high gains for the same, by  showing that the capacity of the SSC goes to one for high coverage depth, allowing highly reliable reconstruction of the string, even for short read lengths $L \approx \log{n}$. In contrast, for the shuffling-sampling channel \cite{DNAStorageISIT} the capacity goes to zero in the same regime, implying we wouldn't be able to recover the set of short strings in the same regime. This result reveals the high capacity gains that would be achieved if data could be stored in a long DNA molecule instead of short-length molecules. We believe exploiting overlaps, which are generally present in reads that are shotgun sequenced, allows for gains in terms of capacity.

 
We further showed that when we shotgun sequence a DNA string that is a codeword from a codebook, we can reduce the minimum required read length by a factor of $2$ and the number of reads sampled by a factor of $\log{n}$, while still admitting arbitrarily close to perfect reconstruction.

DNA synthesis and storage in general admit errors that can cause bit flips and erasures. This has been studied before in the context of the shuffling-sampling channel \cite{noisyshuffling,lenz2018coding,lenz2020achieving}.  This motivates studying a noisy version of the SSC, where the noise can be modeled as concatenated binary symmetric or erasure channels. It is also worthwhile to note that our analysis is restricted to the asymptotic regime as $n \to \infty$. It may be of interest to see the reliability of such systems in the finite blocklength regime.
\newpage

\bibliographystyle{ieeetr}
{\footnotesize 
\bibliography{main}
}

\newpage

\appendices
\section{Proof of Lemma~\ref{lem:coverage}}\label{CovLemma}
 \CoverageLemma*
 
To prove the above result we first compute $E[\Phi]$ as follows:
\begin{align}
    E[\Phi] = \frac{1}{n}\sum_{i=1}^{n}E\left[\mathbf{1}_{\{X_i\text{ is covered by }\Y\}}\right]
    &= \Pr(X_n \text{ is covered by } \Y) \nonumber \\
    &= 1- \Pr\left(X_n \text{ is not covered by } \Y \right) \nonumber \\
    & = 1 - \Pr(X_n \text{ is not covered by } \vec{Y}_i,\;  \forall i \in [1:K]) \nonumber \\
    &\stackrel{(a)}{=} 1 - \Pr(X_n \text{ is not covered by } \vec{Y}_1)^K \nonumber \\
    &= 1 - (1 - \Pr(X_n \text{ is covered by } \vec{Y}_1))^K \nonumber \\
    &\stackrel{(b)}{=} 1- \left(1 - \frac{L}{n}\right)^K \to 1 - e^{-c},
\end{align}
as $n \to \infty$. Here, $(a)$ is due the starting points being picked independently and uniformly at random, and $(b)$ is due to the fact that $X_n$ needs to start in $L$ (contiguous) bits out of $n$ total bits to cover $X_n$.
From the definition of limit, we know that there for every $\epsilon >0$, 
there exists an $N$ such that, for $n \geq N$, we have
\begin{align*}
    |E[\Phi] - (1 - e^{-c})| < \epsilon(1-e^{-c})/2.
\end{align*}
If it also holds that $|\Phi-E[\Phi]|<\epsilon(1-e^{-c})/2$, then by triangle's inequality, we obtain
\begin{align}
   |\Phi - (1 - e^{-c})| \leq |\Phi - E[\Phi]| + |E[\Phi] - (1 - e^{-c})| \leq \epsilon(1-e^{-c}),
\end{align}
for sufficiently large $n$.
Defining $\ep' = \ep(1-e^{-c})/2$ and using Chebyshev's inequality for $n$ large enough, we have

\begin{align}\label{eq:CovChebyshev}
&\Pr\left(\left|\Phi - \left(1 - e^{- c}\right)\right| > \epsilon\left(1 - e^{- c}\right)\right) \nonumber \\
&\leq
\Pr(\left|\Phi - E[\Phi]\right| > \epsilon') 
\leq \frac{\text{Var}(\Phi)}{\epsilon'^2}.
\end{align}
Let's calculate $\text{Var}(\Phi)$. By linearity of covariance, we have
\begin{align}\label{eq:CoverageVarianceCalc}
    \text{Var}(\Phi) = \frac{1}{n^2}\sum_{i,j}\text{Cov}\left(\mathbf{1}_{\{X_i\text{ is covered by }\Y\}}\mathbf{1}_{\{X_j\text{ is covered by }\Y\}}\right).
\end{align}
Now, note that for $(i-j)\mod n > \bar{L}\log{n}$, we can calculate the covariance as 
\begin{align*}
    &\text{Cov}\left(\mathbf{1}_{\{X_i\text{ is covered by }\Y\}}\mathbf{1}_{\{X_j\text{ is covered by }\Y\}}\right) \\
    &= E\left[\mathbf{1}_{\{X_i\text{ is covered by }\Y\}}\mathbf{1}_{\{X_j\text{ is covered by }\Y\}}\right] - E\left[\mathbf{1}_{\{X_i\text{ is covered by }\Y\}}\right]E\left[\mathbf{1}_{\{X_j\text{ is covered by }\Y\}}\right] \\
    &\stackrel{(a)}{\leq} \left(1- \left(1 - \frac{L}{n}\right)^K\right)\left(1- \left(1 - \frac{L}{n}\right)^{K-1}\right) - \left(1- \left(1 - \frac{L}{n}\right)^K\right)^2 \leq 0.
\end{align*}
Here, $(a)$ is because there exists at least one read that has to cover $X_i$ and not $X_j$ when $(i-j)\text{ mod } n > \bar{L}\log{n}$.
Therefore, we can upper bound \eqref{eq:CoverageVarianceCalc} as
\begin{align*}
    \text{Var}(\Phi) &\leq \frac{1}{n^2}\left(\sum_{i=1}^{n} \text{Var}(\mathbf{1}_{\{X_i\text{ is covered by }\Y\}}) + \sum_{i,j:((i-j)\text{ mod } n \leq \bar{L}\log{n})}\text{Cov}\left(\mathbf{1}_{\{X_i\text{ is covered by }\Y\}}\mathbf{1}_{\{X_j\text{ is covered by }\Y\}}\right)\right) \\
    &\leq \frac{1}{n^2}(n + \bar{L}n\log{n}).
\end{align*}
This leads to \eqref{eq:CovChebyshev} being upper bounded as
\begin{align}
    \Pr(\left|\Phi - E[\Phi]\right| > \epsilon' ) 
\leq \frac{\text{Var}(\Phi)}{\epsilon'^2} \leq \left(\frac1n + \frac{\bar{L}\log{n}}{n}\right)\frac1{\epsilon'^{2}}  \to 0,
\end{align}
as $n \to \infty$, which completes the proof.

\section{Proof of Lemma~\ref{lemma:IslandNumbers}}\label{IslandNumberLemma}
\IslandNumbers*

To prove the above lemma, first we note that we can count the number of islands by counting the reads that have no overlap, since any read with no (suffix) overlap must be the last read of an island.
Therefore, we have
\begin{align}
E[K^{\prime}]= \sum_{i=1}^{K}E\left[\mathbf{1}_{\{\vec{Y}_i\text{ has no overlap}\}}\right] = K\Pr(\text{a given read has no overlap}) = K\left(1-\frac{L}{n}\right)^{K-1}.
\end{align}


Note that $\lim_{n \to \infty}
\frac{\log{n}}{n}E[K'] = \frac{c}{\bar{L}}e^{-c}$. Thus, we can say (from the definition of limit) that for any $\epsilon/2 >0$, $|E[K'] - Ke^{-c}| < Ke^{-c}\epsilon/2$, for $n$ large enough. 
Therefore if $|K'-E[K']|<Ke^{-c}\epsilon/2$, then by triangle's inequality
\begin{align*}
    |K' - Ke^{-c}| \leq |K' - E[K']| + |E[K'] - Ke^{-c}| \leq Ke^{-c}\epsilon,
\end{align*}
for $n$ large enough.
Using Chebyshev's inequality, for $n$ large enough we have that
\begin{align}\label{eq:ChebyshevIslands}
& \Pr\left(\middle|K' - Ke^{-c}\middle|\geq \epsilon Ke^{-c}\right) \nonumber \\
&\leq \Pr\left(\middle|K' - E[K']\middle|\geq \epsilon' Ke^{-c}\right) \leq \frac{\text{Var}(K')}{\epsilon'^2K^2e^{-2c}}  \to 0,
\end{align}
as $n \to \infty$ and $\epsilon' = \epsilon/2$. To see this, we can calculate $\text{Var}(K')$ as 
\begin{align*}
   \text{Var}(K') &= \sum_{i,j}\text{Cov}\left(\mathbf{1}_{\{\vec{Y}_i\text{ has no overlap}\}}\mathbf{1}_{\{\vec{Y}_j\text{ has no overlap}\}}\right).
\end{align*}
If $i\neq j$, then we can compute the covariances as follows:
\begin{align*}
& \text{Cov}\left(\mathbf{1}_{\{\vec{Y}_i\text{ has no overlap}\}}\mathbf{1}_{\{\vec{Y}_j\text{ has no overlap}\}}\right) \\
   &= E\left[\mathbf{1}_{\{\vec{Y}_i\text{ has no overlap}\}}\mathbf{1}_{\{\vec{Y}_j\text{ has no overlap}\}}\right] - E\left[\mathbf{1}_{\{\vec{Y}_1\text{ has no overlap}\}}\right]^2 \\
   &\stackrel{(a)}{\leq} \Pr(\vec{Y}_i\text{ and }\vec{Y}_j \text{ don't have overlaps }|\vec{Y}_i \text{ and }\vec{Y}_j \text{ don't overlap with each other}) \\
   &+ \Pr(\vec{Y}_i \text{ and }\vec{Y}_j \text{ overlap with each other}) - \left(1-\frac{L}{n}\right)^{2(K-1)}\\
   &\stackrel{(b)}{\leq} \left(1 - \frac{L}{n}\right)^{2(K-2)} + \frac{k\log{n}}{n} - \left(1-\frac{L}{n}\right)^{2(K-1)} \leq 0,
\end{align*}
for $n$ large enough and a fixed constant $k$. $(a)$ is due the Law of total probability and $(b)$ is because conditioned on the fact that two reads don't overlap with each other, the probability that they don't overlap are independent of one another. We calculate the worst case probability of the rest of the reads overlapping with them. This can be done assuming there $K-2$ reads left at most for each.  Therefore for $n$ large enough, we can say that
\begin{align*}
    \text{Var}(K') \leq \sum_{i=1}^{K}\text{Var}\left(\mathbf{1}_{\{\vec{Y}_i\text{ has no overlap}\}}\right) \leq K.
\end{align*}
Therefore, \eqref{eq:ChebyshevIslands} can be upper bounded as
\begin{align}
    \Pr\left(\middle|K' - Ke^{-c}\middle|\geq \epsilon Ke^{-c}\right) \leq \frac{1}{\ep'^2Ke^{-2c}} \to 0,
\end{align}
as $n \to \infty$.

\section{Proof of Lemma~\ref{lemma:MergeCandidates}}\label{MergeCandidatesLemma}
\MergeCandidates*

We will first prove the first part of the lemma and then the second part of the lemma. Both proofs are similar with some small variations.
We first look at a concentration result on the number of times $\vec{z}$ appears on a length-$n$ i.i.d. Bern$(1/2)$~string $X^n$. We will then use this to prove Lemma~\ref{lemma:MergeCandidates}.
Let $N_{\vec{z}} = \sum_{i=1}^{n} \mathbf{1}_{\{\vec{z} \text{ present starting at the $i$th symbol of $X^n$}\}} := \sum_{i=1}^{n} \mathbf{1}_{\text{F}_i}$. Note that $E\left[N_{\vec{z}}\right]  = n \times 2^{-\gamma(\vec{z}) \log{n}} = n^{1-\gamma(\vec{z})}$.



\begin{lemma}\label{lemma:StringRepeats}
For all $\epsilon > 0$,
\begin{align}
\Pr\left(\left|N_{\vec{z}}-n^{1-\gamma(\vec{z})}\right|\geq \epsilon n^{1-\gamma(\vec{z})}\right) \leq (2\Bar{L} \log{n})e^{-(n^{\ep})\epsilon^2/(2\Bar{L}\log{n})}.
\end{align}
\end{lemma}
\begin{proof}
We can rewrite $N_{\vec{z}}$ as
\aln{
    N_{\vec{z}} &= \sum_{i=1}^{n} \mathbf{1}_{F_i} =
    \underbrace{
    \sum_{t=0}^{\frac{n}{\gamma(\vec z) \log n} - 1}\mathbf{1}_{F_{1+t \gamma(\vec z) \log n}}}_{N_{1,\vec{z}}} + \underbrace{
    \sum_{t=0}^{\frac{n}{\gamma(\vec z) \log n} - 1}\mathbf{1}_{F_{2+t \gamma(\vec z) \log n}}}_{N_{2,\vec{z}}} + \dots +
    \underbrace{
    \sum_{t=0}^{\frac{n}{\gamma(\vec z) \log n} - 1}\mathbf{1}_{F_{\gamma(\vec z) \log n +t \gamma(\vec z) \log n}}}_{N_{\gamma(\vec z) \log n ,\vec{z}}}.
}
Notice that each summation 
deals with starting locations that are at least $\gamma(\vec z)\log n$ symbols apart, and the resulting random variables $\1_{F_i}$ are independent.
Now note that by triangle inequality if each of 
\begin{align}
 \left|N_{i,\vec{z}} - \frac{n^{1-\gamma(\vec{z})}}{\gamma(\vec{z})\log{n}}\right|\leq \epsilon \frac{n^{1-\gamma(\vec{z})}}{\gamma(\vec{z})\log{n}}
\end{align} 
then $\left|N_{\vec{z}}-n^{1-\gamma(\vec{z})}\right|\leq \epsilon n^{1-\gamma(\vec{z})}$. Employing the union bound, we can say that
\begin{align}
    \Pr\left(\left|N_{\vec{z}}-n^{1-\gamma(\vec{z})}\right|\geq \epsilon n^{1-\gamma(\vec{z})}\right) 
    &\leq \sum_{i=1}^{\gamma(\vec{z}){\log{n}}}\Pr\left(\left|N_{i,\vec{z}}-\frac{n^{1-\gamma(\vec{z})}}{\gamma(\vec{z})\log{n}}\right|\geq \epsilon \frac{n^{1-\gamma(\vec{z})}}{\gamma(\vec{z})\log{n}}\right) \nonumber \\
    &\leq (\Bar{L}\log{n})\Pr\left(\left|N_{1,\vec{z}}-\frac{n^{1-\gamma(\vec{z})}}{\gamma(\vec{z})\log{n}}\right|\geq \epsilon \frac{n^{1-\gamma(\vec{z})}}{\gamma(\vec{z})\log{n}}\right) \nonumber \\
    &\stackrel{(a)}{\leq} (2\Bar{L} \log{n}) e^{\frac{-n\times n^{-\gamma(\vec{z})}\epsilon^2}{2\gamma(\vec{z})(\log{n})(1-n^{-\gamma(\vec{z})})}} \nonumber \\
    &\leq (2\Bar{L} \log{n})e^{-n^{1-\gamma(\vec{z})}\epsilon^2/(2\Bar{L}\log{n})} \nonumber\\
    &\stackrel{(b)}{\leq} (2\Bar{L} \log{n})e^{-(n^{\ep})\epsilon^2/(2\Bar{L}\log{n})}
\end{align}
where step $(b)$ holds for all $\gamma(\vec z) \leq 1- \ep$.
Here $(a)$ is due to the following form of Hoeffding's inequality:
\begin{lemma}\label{lemma:ModifiedChernoff}
For i.i.d. Bernoulli $X_1,X_2\dots X_n$ with parameter $p$,
\begin{align}
    \Pr\left(\left|\frac{1}{n}\sum_{i=1}^{n}X_i - p \right| \geq \epsilon p\right) \leq 2e^{-nD(p+\epsilon||p)} \leq 2e^{-np\epsilon^2/(1-p)}
\end{align}
\end{lemma}
This ends the proof of Lemma~\ref{lemma:StringRepeats}.
\end{proof}

Now given the string $X^n$, define $\mathcal{D} := \{ \text{Starting point indices of $\vec{z}$ in $X^n$}\}$. Therefore $M_{\vec{z}} = \sum_{i=1}^K \mathbf{1}_{\{\text{$i$th read has starting point in $\mathcal{D}$}\}}$. Since the generation of these reads is independent and uniformly at random, the random variables in the summation are independent of each other (given ${\mathcal{D}}$). Now we can say
\begin{align}\label{eq:Mzanalysis}
    \Pr&\left(\bigcup_{z 
    \in \mathcal{Z},\gamma(\vec z) < 1}\left\{\left|M_{\vec{z}}-K n^{-\gamma(\vec z)}\right|\geq  \ep K n^{-\gamma(\vec z)}\right\}\right)  \nonumber \\
    &\stackrel{(a)}{\leq} {d}n\max_{\vec{z}\in\mathcal{Z},\gamma(\vec z)< 1}\Pr\left(\left\{\left|M_{\vec{z}}-K n^{-\gamma(\vec z)}\right|\geq  \ep K n^{-\gamma(\vec z)}\right\}\right) \nonumber \\
    &\stackrel{(b)}{\leq} {d}n\max_{\vec{z}\in\mathcal{Z},\gamma(\vec z)< 1}\left(\Pr\middle(\{\middle|M_{\vec{z}}-K n^{-\gamma(\vec z)}\middle|\geq  \ep K n^{-\gamma(\vec z)}\middle\}\middle|\middle\{|N_{\vec{z}}-n^{1-\gamma(\vec{z})}|\middle\}\leq \epsilon n^{1-\gamma(\vec{z})}\middle)\right) \nonumber \\
    &+ \Pr\left(\middle|N_{\vec{z}}-n^{1-\gamma(\vec{z})}\middle|\geq \epsilon n^{1-\gamma(\vec{z})}\middle)\right) \nonumber\\
    &\stackrel{(c)}{\leq} {d}n\left(2e^{-\frac{cn^{\ep}}{2\Bar{L}\log{n}}\epsilon^2} + (2\Bar{L} \log{n})e^{-n^{\ep}\epsilon^2/2\Bar{L}\log{n}}\right) \to 0,
\end{align}
as $n \to \infty$. Here $(a)$ is due to the union bound on the number of $\vec{z}$ for $\gamma(\vec z) \leq 1-\epsilon$ (which is $\leq {d}n$, {for some fixed constant ${d}$, since there are at most $\sum_{\gamma \leq 1-\ep}2^{\gamma\log{n}} = \sum_{\gamma \leq 1-\ep}n^{\gamma}$ vectors of that size, and for $\gamma \leq1-\ep$, $\sum_{\gamma \leq 1-\ep}n^{\gamma} \leq n^{1-\ep}\log{n} \leq dn$, for $n$ large enough),}
$(b)$ is from the law of total probability and $(c)$ is due to the Hoeffding's inequality. 
To see the last step note that $\mathbf{1}_{\{\text{$i$th read starts in $\mathcal{D}$}\}}$ is a Bern($p$) RV. When conditioned on the event $\left|N_{\vec{z}}-n^{1-\gamma(\vec{z})}\right|\leq \epsilon n^{1-\gamma(\vec{z})}$, $p \in [(1-\epsilon)n^{-\gamma(\vec{z})},(1+\epsilon)n^{-\gamma(\vec{z})}]$. We can apply Hoeffding's inequality assuming the worst case of $p = (1+\epsilon)n^{-\gamma(\vec{z})}$. {Now since $\gamma \leq 1 - \ep$ in the above case, this implies that $np \leq n\times(1+\epsilon)n^{-\gamma(\vec{z})} \leq (1+\ep)n^{\ep}$.}
Now apply this bound on $np$ to  Lemma~\ref{lemma:ModifiedChernoff}, to obtain the result.

Now for the second part of the proof we employ a similar trick as before. Using the same split on $N_{\vec z}$ as before we have that
\begin{align}\label{eq:prelimpart2apen3}
    &\Pr\left(N_{\vec{z}} \geq \frac{L\times n^{\ep}}{c}\right) \stackrel{(a)}{\leq} (\bar L \log{n})\Pr\left(N_{1,\vec z} \geq \frac{n^{\ep}}{c}\right) \\
    &=  (\bar L \log{n})\Pr\left(N_{1,\vec z} - \frac{n^{1-\gamma}}{L}\geq \frac{n^{\ep}}{c} - \frac{n^{1-\gamma}}{L}\right) \leq (\bar L \log{n})\exp{\left(-\frac{\left(n^{\ep}/c - n^{1-\gamma}/L\right)^2}{2n^{1-\gamma}(1-n^{-\gamma})}\right)} \\
    &\leq (\bar L \log{n})\exp{\left(-n^{(2\ep-(1-\gamma'))}\frac{\left(1/c - n^{1-\gamma-\ep}/L\right)^2}{2}\right)},
\end{align}
where $(a)$ is due to the union bound, specifically because $\gamma \leq \bar{L}$, which implies that there are at most $\bar{L}\log{n}$ terms. Here we use the following version of Hoeffding's inequality:
\begin{lemma}
For i.i.d. Bernoulli $X_1,X_2\dots X_n$ with parameter $p$
\begin{align}
    \Pr\left(\sum_{i=1}^{n}X_i - np \geq x\right)  \leq e^{-x^2/(2np(1-p))}
\end{align}
\end{lemma}

{We now use \eqref{eq:prelimpart2apen3}, to condition the event $\{M_{\vec{z}}\geq n^\ep\}$ on $N_{\vec{z}} \leq (Ln^{\ep})/c$} and follow an analysis similar to \eqref{eq:Mzanalysis}, to arrive at
\begin{align*}
    \Pr\left(\bigcup_{z 
    \in \mathcal{Z}:\gamma(\vec{z})> 1-\ep}\left\{M_{\vec{z}} \geq n^{\ep}\right\}\right) \to 0,
\end{align*}
as $n \to \infty$.

\section{Proof of Lemma~\ref{lem:OverlapNumber}}\label{OverlapNumber}
\OverlapNumber*
\begin{proof}
We use Chebyshev's inequality to see that
\begin{align}
    \text{Pr}(|G(\gamma) - \bar{G}(\gamma)|\geq \epsilon \bar{G}(\gamma)) \leq \frac{\text{Var}(G(\gamma))}{\epsilon^2\bar{G}(\gamma)^2}.
\end{align}
This can be rewritten using the fact that
\begin{align}
    \frac{\text{Var}(G(\gamma))}{\bar{G}(\gamma)^2} = \frac{E[G(\gamma)^2]}{\bar{G}(\gamma)^2} - 1.
\end{align}
We also have that
\begin{align}
    E[G(\gamma)^2] &= E\left[\left(\sum_{i=1}^{K}G(\gamma)_i\right)^2\right] \nonumber\\
    &= \left(\sum_{i=1}^{K}E\left[G(\gamma)_i^2\right]\right) + \sum_{i\neq j}E[G(\gamma)_iG(\gamma)_j] \nonumber\\
    &\leq KE[G(\gamma)_1^2] + K^2E[G(\gamma)_1G(\gamma)_2] \nonumber \\
    & = KE[G(\gamma)_1] + K^2E[G(\gamma)_1G(\gamma)_2] = \bar{G}(\gamma) + K^2E[G(\gamma)_1G(\gamma)_2].
\end{align}
{Recalling that the definition of ``overlap", implies an overlap only between successive reads (Definition~\ref{def:overlap})} and focusing on the second term in the above equation, we have
\begin{align}
    &E[G(\gamma)_1G(\gamma)_2] = \Pr(Y_1,Y_2\text{ both have overlap size $\gamma\log{n}$})\nonumber \\
    &\leq \Pr(Y_1,Y_2\text{ both have overlap size $\gamma\log{n}$}|Y_1 \text{ has no overlap with } Y_2) \nonumber \\
    & \quad + \Pr(Y_1 \text{ has an overlap with } Y_2)
    \nonumber \\
    &\stackrel{(a)}{\leq} \Pr(Y_1,Y_2\text{ both have overlap size $\gamma\log{n}$}|Y_1 \text{ has no overlap with } Y_2) + \frac{k\log{n}}{n} \nonumber \\
    &\stackrel{(b)}{\leq} {\frac{\bar{G}(\gamma)^2}{K^2}} + k\frac{\log{n}}{n},
\end{align}
 where $k$ is a finite constant. 
 Here $(a)$ is because two reads with lengths $\bar{L}\log{n}$ have an overlap if their starting points $T_1,T_2$ are such that $T_2-T_1 \leq L$ and $(b)$ is because 
 the conditioning constrains the starting points to not be such that $T_2-T_1 \leq L$, 
 implying that the events $\{Y_i \text{ has an overlap of size } \gamma \log n\}$ for $i=1,2$ are conditionally independent and with probability at most $E[G(\gamma)_1] = \bar G(\gamma)/K$.
 Therefore we have
\begin{align}
    E[G(\gamma)^2] \leq \bar{G}(\gamma) + \bar{G}(\gamma)^2 + K^2 k \frac{\log{n}}{n}.
\end{align}
Hence,
\begin{align}
    \frac{\text{Var}(G(\gamma))}{\bar{G}(\gamma)^2} &= \frac{E[G(\gamma)^2]}{\bar{G}(\gamma)^2} - 1 \leq \frac{1}{\bar{G}(\gamma)} + 1 -1 +\frac{K^2k\log{n}}{n\bar{G}(\gamma)^2} = \frac{1}{\bar{G}(\gamma)} + \frac{K^2k\log{n}}{n\Theta\left(n^2/\log^4{n}\right)}
    \nonumber \\
    &\stackrel{(a)}{\leq} \frac{n}{(K-1)^2\left(1-\frac{(\Bar{L}-\gamma)\log{n}+1}{n}\right)^{(K-2)}} + {\Theta(n^{-1}\log^3{n})} \nonumber \\
    & \leq \frac{n}{(K-1)^2\left(1-\frac{\Bar{L}\log{n}+1}{n}\right)^{(K-2)}} + {\Theta(n^{-1}\log^3{n})} \nonumber \\
    &= {\Theta\left(\frac{\log^3{n}}{n}\right)},
\end{align}
where (a) is due to the lower bound in \eqref{eq:ExpectedValueUpperLower}, followed by further noting that $K-1 \leq K$. 
Now due to the union bound
\begin{align*}
    \Pr\left(\bigcup_{\gamma \in \Gamma} \left\{ |G(\gamma) - \bar{G}(\gamma)|\geq \epsilon  \bar{G}(\gamma) \right\} \right) \leq \frac{\Bar{L}\log{n}}{\epsilon^2}\cdot {\Theta\left(\frac{\log^3{n}}{n}\right)} \to 0,
\end{align*}
as $n \to \infty.$
\end{proof}

\section{Proof of Lemma~\ref{lem:CostSumEvaluation}}\label{CostSumLem}
\CostSumEvaluation*
From the definition of $\overline{CI}_n$, 
\begin{align}
    \overline{CI}_n &= P \times \prod_{\gamma \leq 1-\ep} B_{3}(\gamma)^{B_4(\gamma)}\times\prod_{\gamma > 1-\ep} n^{\ep B_4(\gamma)} \nonumber \\
    &= (L+1)^{K} \times (1+\ep)^{(1-\ep) \bar{L}\log{n}}n^{(1+\ep)\sum_{\gamma \leq 1 - \ep}(1-\gamma) \bar{G}(\gamma)}\times n^{\ep(1+\ep)\sum_{\gamma > 1-\ep}\bar{G}(\gamma)}.
\end{align}
Hence, we have that 
\begin{align}\label{eq:costfirst}
    &\lim_{n \to \infty}\frac1n\log{\overline{CI}_n}\nonumber \\
    & = \lim_{n \to \infty}\left(\frac{K\log{(L+1)}}{n} + \frac{\bar{L} \log{n}}{n}(1-\ep)\log{(1+\ep)} \right. \nonumber \\
    & \quad \quad \quad \quad \quad \left. +\frac{\log{n}}{n} (1+\ep) \sum_{\gamma \leq 1-\ep}(1-\gamma)\bar{G}(\gamma) + \frac{\ep(1+\ep)\log{n}}{n}\sum_{\gamma > 1-\ep}\bar{G}(\gamma)\right) \nonumber \\
    &= \lim_{n\to\infty} \frac{\log{n}}{n} (1+\ep) \sum_{\gamma \leq 1-\ep}(1-\gamma)\bar{G}(\gamma)
     + \ep (1+\ep) \lim_{n \to \infty}\frac{\log{n}}{n}\sum_{\gamma > 1-\ep}\bar{G}(\gamma).
\end{align}
Notice that second term in \eqref{eq:costfirst} satisfies
\begin{align*}
    0 \leq \ep(1+\ep)\lim_{n \to \infty}\frac{\log{n}}{n}\sum_{\gamma > 1-\ep}\bar{G}(\gamma) \leq \ep(1+\ep) \lim_{n\to\infty}\frac{\log{n}}{n}K = \ep(1+\ep) \frac{c}{\bar{L}} \to 0,
\end{align*}
as $\ep \to 0$.

Now we look at the first term in \eqref{eq:costfirst}. 
Before we proceed to evaluate the required summation, we need to calculate  $\bar{G}(\gamma)$. Calculating this exactly is difficult, but we can find upper and lower bounds that asymptotically converge. 
We start by noticing that
\begin{align*}
    \bar{G}(\gamma) &= K\Pr(X_1 \text{ has an overlap of size } \gamma\log{n}) \\
    &= K\Pr(X_1 \text{ has an overlap of size } \gamma\log{n}|X_1 \text{ starts from position } 1) \\
    &= K\Pr(\text{min start location of }(X_2,\dots,X_K) = (\Bar{L} - \gamma)\log{n}|X_1 \text{ starts from } 1).
\end{align*}
Let's look at $\Pr(\text{min start location of }(X_2,\dots,X_K) = (\Bar{L} - \gamma)\log{n}|X_1 \text{ starts from } 1)$. 
We can upper and lower bound this probability by 
forcing one read to start at position $(\Bar{L} - \gamma)\log{n}$ and all others to start at position $(\Bar{L} - \gamma)\log{n}$ or higher
(which will lead to double counting, and thus an upper bound). 
For the lower bound we can assume that exactly one of the reads starts at $(\Bar{L} - \gamma)\log{n}$ and the rest start at positions strictly greater than $(\Bar{L} - \gamma)\log{n}$. Thus we get
\begin{align}\label{eq:ExpectedValueUpperLower}
  K(K-1)\times\frac{1}{n}\left(1-\frac{(\Bar{L} - \gamma)\log{n}+1}{n}\right)^{K-2} \leq \bar{G}(\gamma) \leq K(K-1)\times\frac{1}{n}\left(1-\frac{(\Bar{L} - \gamma)\log{n}}{n}\right)^{K-2}.
\end{align}
Pick $\Delta_n =\gamma + \frac{1}{\log{n}}-\gamma$. Now we can say that
\begin{align}\label{eq:Complexsumfirststep}
   \lim_{n \to \infty}& \frac{\log{n}}{n}\sum_{\gamma \leq 1-\ep} (1-\gamma)\bar{G}(\gamma)
= \lim_{n \to \infty}\sum_{\gamma \leq 1-\ep}\left((1-\gamma)\times\frac{\log{n}}{n}\times\frac{\bar{G}(\gamma)}{\Delta_n}\Delta_n\right) \nonumber \\
&\stackrel{(a)}{=} \int_{0}^{1-\ep}(1-\gamma)\lim_{n \to \infty}\left(\bar{G}(\gamma)\frac{(\log{n})^2}{n}\right)d\gamma,
\end{align}
where $(a)$ follows from the definition of Riemann integration.
Let us evaluate the limit inside the integral first. 
Applying the Sandwich theorem to \eqref{eq:ExpectedValueUpperLower}, we have that
\begin{align*}
  \lim_{n \to \infty}\bar{G}(\gamma)\frac{(\log{n})^2}{n} = \frac{c^2}{\Bar{L}^2}\exp\left(-c\left(1-\frac{\gamma}{\Bar{L}}\right)\right).
\end{align*}
Therefore from \eqref{eq:Complexsumfirststep} 
we have that
\begin{align*}
\lim_{n \to \infty}& \frac{\log{n}}{n}\sum_{\gamma \leq 1-\ep} (1-\gamma)\bar{G}(\gamma)
= \int_{0}^{1-\ep}(1-\gamma)\lim_{n \to \infty}\left(\bar{G}(\gamma)\frac{(\log{n})^2}{n}\right)d\gamma \nonumber \\
    &= \frac{c^2}{\Bar{L}^2} \int_{0}^{1-\ep}(1-\gamma)\exp{\left(-c\left(1-\frac{\gamma}{\Bar{L}}\right)\right)}d\gamma \\
    & \leq \frac{c^2}{\Bar{L}^2} \int_{0}^{1}(1-\gamma)\exp{\left(-c\left(1-\frac{\gamma}{\Bar{L}}\right)\right)}d\gamma \\
    &= \frac{c^2}{\Bar{L}^2}e^{-c\left(1-\frac{1}{\Bar{L}}\right)}\int_{0}^{1}ze^{-\frac{c}{\Bar{L}}z}dz \\
    &= e^{-c\left(1-\frac{1}{\Bar{L}}\right)}\left(1-\left(\frac{c}{\Bar{L}}+1\right)e^{-\left(\frac{c}{\Bar{L}}\right)}\right) = \left(e^{-c\left(1-\frac{1}{\Bar{L}}\right)}-\left(\frac{c}{\Bar{L}}+1\right)e^{-c} \right),
\end{align*}
where we used the substitution $z=1-\gamma$.
Finally, plugging everything back into \eqref{eq:costfirst}, we obtain
\aln{
    \lim_{n \to \infty}\frac1n\log{\overline{CI}_n} & \leq 
    (1+\ep) \left(e^{-c\left(1-\frac{1}{\Bar{L}}\right)}-\left(\frac{c}{\Bar{L}}+1\right)e^{-c} \right) + 
    \ep(1+\ep)\frac{c}{\bar L} \\
    & = e^{-c\left(1-\frac{1}{\Bar{L}}\right)}-\left(\frac{c}{\Bar{L}}+1\right)e^{-c}  + f(\ep),
}
where $f(\ep) \to 0$ as $\ep \to 0$, concluding the proof.

\section{Proof of Lemma~\ref{lem:MaxIslandsBound}}\label{MaxIslandBound}
\MaxIslandsBound*

We first upper bound the $\Pr\left(D > \gamma_0 \log{n}\right)$ as
\begin{align}\label{eq:Lemma3UnionBd}
    &\Pr\left(D > \gamma_0 \log{n}\right) \nonumber \\
    &\leq K \Pr\left(\text{No. of samples in a given island} > \gamma_0 \log{n}\right)
\end{align}
and then show that this upper bound tends to zero as $n \to \infty.$
We first define a few terms.
Let the sequence $\vec{Y}_1,\vec{Y}_2,\dots$ represent the reads (in an ordered fashion) in the given island. We define $U_i := T_{i+1}-T_i$ with $T^K$ being the vector of ordered starting locations, as the separation between read $\vec{Y}_i$ and $\vec{Y}_{i+1}$. We can thus think of the sampling scheme as a random process that picks starting locations $T^K$ with inter-arrival times $U_1,U_2, \dots, U_{K-1}$. For convenience, we define $U_K = ((T_1 - T_K) \mod n) + 1$ (this is to capture the wrap around nature of the reads). Note that $\sum_{i=1}^KU_i = n$.

Thus we calculate $\Pr\left(\text{A given read has an overlap}\right)$ as follows,
\begin{align}\label{eq:ReadsWithoutOverlaps}
     \Pr\left(\text{A given read has an overlap}\right) &= 1 - \Pr\left(\text{A given read has no overlap}\right) \nonumber \\
    &=1 - \Pr(\text{A particular read does not overlap with the given read})^{K-1} \nonumber \\
    &= 1 - \left(1-\frac{L}{n}\right)^{K-1}.
\end{align}
We now use the above quantity to calculate \eqref{eq:Lemma3UnionBd}.
We note that
\begin{align}
    &\Pr\left(\text{No. of samples in a given island} > \gamma_0 \log{n}\right) \nonumber \\
    &= \Pr\left(\text{No. of samples in a given island} > \gamma_0 \log{n} \middle|\text{given island starts from $\vec{Y}_1$} \right)  \nonumber\\
    &{=}\Pr(U_1\leq L,U_2 \leq L,\dots, U_{\gamma_0\log{n}-1} \leq L) \nonumber \\
    &= \Pr(U_1 \leq L)\Pr(U_2 \leq L|U_1 \leq L)\dots\Pr( U_{\gamma_0\log{n}-1}|U_1\leq L,U_2 \leq L\dots U_{\gamma_0\log{n}-2} \leq L) \nonumber \\
    &\stackrel{(a)}{\leq} \Pr(U_1 \leq L)\Pr(U_2 \leq L)\dots\Pr( U_{\gamma_0\log{n}-1} \leq L) \nonumber \\
    &= \left(1 - \left(1-\frac{L}{n}\right)^{K-1}\right)^{\gamma_0\log{n}-1},
\end{align}
where $(a)$ is because the $U_i,i\in[1:K]$ are negatively associated (because the larger one $U_i$ is, the less room there is on $x^n$ for the other $U_j$'s). 
For negatively associated random variables we know that \cite{negassociated}
\begin{align}
    \Pr(U_i\leq L:i \in [1:K]) \leq \Pi_{i=1}^{K}\Pr(U_i \leq L)
\end{align}
Therefore \eqref{eq:Lemma3UnionBd} can be upper bounded as
\begin{align}\label{eq:MaxBoundUpperBound}
    \Pr\left(D \geq \gamma_0 \log{n}\right) &\leq K \left(1 - \left(1-\frac{L}{n}\right)^{K-1}\right)^{\gamma_0\log{n}-1}  \nonumber\\
    &=\frac{c\times2^{\log{n} + (\gamma_0\log{n}-1)\log\left(1 - \left(1-\frac{L}{n}\right)^{K-1}\right)}}{\bar{L}\log{n}} \nonumber\\
    &= \frac{c\times2^{\log{n}\left(1+\gamma_0\log\left(1 - \left(1-\frac{L}{n}\right)^{K-1}\right)\right) - \log\left(1 - \left(1-\frac{L}{n}\right)^{K-1}\right)}}{\bar{L}\log{n}}.
\end{align}
Now as long as
\begin{align*}
    &\lim_{n \to \infty} \left(1+\gamma_0\log\left(1 - \left(1-\frac{L}{n}\right)^{K-1}\right)\right) < 0 \text{ or } \\
    &\gamma_0 > \frac{-1}{\log{(1-e^{-c})}},
\end{align*}
\eqref{eq:MaxBoundUpperBound} $\to 0$ as $n \to \infty$.
\section{Finiteness of Island lengths}\label{FiniteIslandLength}
We are required to prove that $(a)$ $\lim_{n \to \infty} \frac{\log{n}}{E[N_1]} \in (0,\infty)$ and $(b)$ $E[N_1^2/(\log{n})^2]$ is finite and bounded.
We note that $N_i = \sum_{i=1}^{J}Z_i$, where $J$ is the random variable which indicates the number of reads in an island and $Z_i$ is the length of the reads after removing the overlapping part of the read.
To see $(a)$ note that there are $K^{\prime}$ islands. Define $J_1,J_2,\dots J_{K^{\prime}}$ as the number of reads in each of these islands. Note that $J$ and $J_i$ are identically distributed for all $i$.
Therefore we can say that
\begin{align*}
    K = \sum_{i=1}^{K^{\prime}}J_i.
\end{align*}
Note that $K'$ is a stopping time with respect to $J_1,J_2,\dots$. 
This implies that $E[K']E[J] = K$. Also note that $J$ is a stopping time with respect to $Z_1,Z_2\dots$. Therefore we can say
\begin{align*}
    \lim_{n \to \infty} \frac{\log{n}}{E[N_1]} &= \lim_{n \to \infty}  \frac{\log{n}}{E[J]E[Z_1]} =\lim_{n \to \infty} \frac{E[K']\log{n}}{KE[Z_1]} \\
    &{= \left(\lim_{n \to \infty}\frac{E[K']}{K}\right)\left(\lim_{n \to \infty}\frac{\log{n}}{E[Z_1]}\right) \stackrel{(a)}{=} \frac{ce^{-c}}{\bar{L}}d\in (0,\infty),}
\end{align*}
{where $d$ is a fixed constant, since $E[Z_1] \sim \Theta(\log{n}))$.}
 To handle $(b)$ we note the following
\begin{align}\label{eq:SecondMomentBound}
    E[N_i^2] &= E\left[\left(\sum_{i=1}^{J}Z_i\right)^2\right] {\leq E\left[\left(\sum_{i=1}^{J}L\right)^2\right] = E[J^2L^2] =  E[J^2]L^2,}
\end{align}
since each of $Z_i \leq L$ for all $i$.
Let us look at the distribution of $J$.
We notice that
\begin{align*}
    &\Pr(J = i) = \Pr(J=i|\text{Island starts from first read}) \\
    & = \Pr(U_1\leq L, \dots, U_{i-1} \leq L ,U_i > L) \nonumber \\
    &\leq \Pr(U_1\leq L, \dots, U_{i-1} \leq L)\nonumber \\
    &\stackrel{(a)}{\leq} \Pr(U_1 \leq L)\Pr(U_2 \leq L)\dots \Pr(U_{i-1} \leq L) \nonumber \\
    &\leq \left(1-\left(1-\frac{L}{n}\right)^{K-1}\right)\times\left(1-\left(1-\frac{L}{n}\right)^{K-1}\right)\dots\left(1-\left(1-\frac{L}{n}\right)^{K-1}\right) \\
    &= \left(1-\left(1-\frac{L}{n}\right)^{K-1}\right)^{i-1},
\end{align*}
where $(a)$ is due to the fact that $U_i,i\in[1:K]$ are negatively associated.
Now we can upper bound $E[J^2]$ as
\begin{align*}
    &E[J^2] = \sum_{i=1}^{K}i^2\Pr(J=i) \\
    &\leq \sum_{i=1}^{K}i^2\left(1-\left(1-\frac{L}{n}\right)^{K-1}\right)^{i-1} \\
    &\stackrel{(a)}{\leq} \sum_{i=1}^{K}i^2\left(1 - \exp{\left(\frac{-\frac{LK}{n}}{\left(1 - \frac{L}{n}\right)}\right)}\right)^{i-1} \stackrel{(b)}{\leq} \sum_{i=1}^{K}i^2\left(1 - \exp{\left(-\frac{c}{\left(1 - \frac{\Bar{L}\log{2}}{2}\right)}\right)}\right)^{i-1},
\end{align*}
where $(a)$ is since $1-x \geq e^{\frac{-x}{1-x}}$ and $(b)$ is due to the max of $\log{n}/n$ being $\log{2}/2$ for $n \in \mathbb{N}$. 

The above series converges to a finite value. To see this, note that this summation is of the form $\sum_{i=1}^K i^2 \alpha^i$, for $\alpha \in (0,1)$, which converges by the root test.
Therefore $E[J^2]$ has to have a finite bound (let this be $M$). This implies \eqref{eq:SecondMomentBound} can be upper bounded as
\begin{align*}
    E[N_1^2] \leq M \Bar{L^2}(\log{n})^2.
\end{align*}

This implies that $E[N_1^2/(\log{n})^2]$ is bounded and finite.

\section{Proof of equation~\ref{eq:ApparantIslandLength}}\label{ApparantIslandLengthAppendix}
We aim to prove \eqref{eq:ApparantIslandLength} here. First note that $E[K^{\prime\prime}]$ can be calculated as 
\begin{align}\label{eq:ApparantIslandNumber}
    E[K^{\prime\prime}] &= K\Pr(\text{Read has overlap size }\leq \delta\log{n}) \nonumber \\
    & = K\Pr(U_1\geq L-\delta\log{n}) = K\left(1-\frac{L-\delta\log{n}}{n}\right)^{K-1}.
\end{align}
It is easy to see that $\lim_{n \to \infty} \frac{\log{n}}{n}E[K^{\prime\prime}] = \frac{c}{\Bar{L}}e^{-c\sigma}$.
Note that the formation of apparent islands can be interpreted as equivalently finding real islands with read lengths truncated by $\delta\log{n}$ (except for the last read in the island). The quantity $\lim_{n \to \infty}\frac{1}{n}E[K^{\prime\prime}]E[N_1^{\delta}]$ is just the coverage of this modified expression. But since the last read is still size $L$, we can think of that read in two parts, one part which contributes $L - \delta\log{n}$ and the other $\delta\log{n}$. Therefore the total coverage is
\begin{align}\label{eq:ApparantCoverage}
    1-e^{-c\sigma} + c(1-\sigma)e^{-c\sigma}.
\end{align}
This is because the reads of length $L- \delta\log{n}$ contribute to a coverage of $1-e^{-c\sigma}$ (the effective coverage depth is shortened). Now the additional $\delta\log{n}$ contribute individually an extra length for each island, but this only needs to be added for the last read. Since there are $K''$ islands, the cumulative contribution is $\lim_{n \to \infty}\frac{1}{n}E[K''](\delta\log{n}) = c(1-\sigma)e^{-c\sigma}$.

Therefore from \eqref{eq:ApparantIslandNumber} and \eqref{eq:ApparantCoverage}, we can say that
\begin{align}
    \lim_{n \to \infty}\left(\frac{1}{n}E[K^{\prime\prime}]E[N_1] - E[K
    ^{\prime\prime}]\frac{\log{n}}{n} \right) = (1-e^{-c\sigma}) + c(1-\sigma)e^{-c\sigma} - \frac{c}{\Bar{L}}e^{-c\sigma}
\end{align}
\end{document}